\pdfoutput=1
\documentclass[12pt,fleqn]{article}

\usepackage{ifthen}
\newboolean{FULL}
\setboolean{FULL}{true}

\usepackage{amsmath,amssymb,amsthm,mathtools,thm-restate}
\usepackage{fullpage}

\usepackage{comment}

\usepackage[T1]{fontenc}
\usepackage{fouriernc}
\usepackage[cal=pxtx,scr=boondox,bb=txof,frak=pxtx]{mathalpha}
\let\newmathbb\mathbb
\AtBeginDocument{%
    \let\mathbb\relax
    \newcommand{\mathbb}[1]{\bm{\newmathbb{#1}}}
}
\usepackage[style=alphabetic,natbib=true,sortcites=true,backref=true,maxnames=99,maxalphanames=99]{biblatex}

\DefineBibliographyStrings{english}{%
  backrefpage = {$\Lsh$~p.},
  backrefpages = {$\Lsh$~pp.},
}
\usepackage[colorlinks=true,citecolor=blue!65!black,linkcolor=red!75!black]{hyperref}
\usepackage[nameinlink]{cleveref}

\usepackage[basic,sanserif]{complexity}

\usepackage{booktabs,colortbl}
\usepackage{diagbox,enumitem}
\usepackage{framed,fancybox,ascmac}
\usepackage{tabularx}

\usepackage{comment}
\usepackage{url}

\usepackage[colorinlistoftodos,prependcaption]{todonotes}

\usepackage{xspace}
\usepackage{bm}

\usepackage{lineno}

\usepackage{tikz}
\usetikzlibrary{positioning}
\usetikzlibrary{decorations.markings}
\usetikzlibrary{shapes.geometric}
\usetikzlibrary{arrows.meta}
\usetikzlibrary{arrows,automata,decorations.pathmorphing,fit,positioning,arrows.meta,calc}

\crefname{theorem}{Theorem}{Theorems}
\crefname{proposition}{Proposition}{Propositions}
\crefname{lemma}{Lemma}{Lemmas}
\crefname{claim}{Claim}{Claims}
\crefname{corollary}{Corollary}{Corollaries}
\crefname{remark}{Remark}{Remarks}
\crefname{observation}{Observation}{Observations}
\crefname{hypothesis}{Hypothesis}{Hypotheses}
\crefname{definition}{Definition}{Definitions}
\crefname{problem}{Problem}{Problems}
\crefname{example}{Example}{Examples}

\crefname{appendix}{Appendix}{Appendices}
\crefname{section}{Section}{Sections}
\crefname{equation}{Eq.}{Eqs.}
\crefname{figure}{Figure}{Figures}
\crefname{table}{Table}{Tables}

\crefname{algorithm}{Algorithm}{Algorithms}

\usepackage[ruled]{algorithm}
\usepackage{algpseudocodex}
\algnewcommand{\algorithmicand}{\textbf{ and }}
\algnewcommand{\algorithmicor}{\textbf{ or }}
\algnewcommand{\algorithmicto}{\textbf{ to }}

\algrenewcommand\textproc{\textsl}

\renewcommand{\geq}{\geqslant}
\renewcommand{\leq}{\leqslant}
\renewcommand{\phi}{\varphi}
\renewcommand{\epsilon}{\varepsilon}
\renewcommand{\bar}{\overline}

\renewcommand{\tilde}{\widetilde}

\newcommand{\prb}[1]{\textup{\textsc{#1}}\xspace}
\newcommand{\nth}[1]{#1\textsuperscript{th}\xspace}

\renewcommand{\vec}[1]{\mathbf{\bm{#1}}}

\newcommand{\reco}{\leftrightsquigarrow}

\DeclareMathOperator*{\argmax}{argmax}

\DeclareMathOperator{\bigO}{\mathcal{O}}

\DeclareMathOperator{\opt}{\mathsf{opt}}
\DeclareMathOperator{\cost}{\mathsf{cost}}
\let\poly\relax\DeclareMathOperator{\poly}{poly}
\let\polylog\relax\DeclareMathOperator{\polylog}{polylog}
\DeclareMathOperator{\polyloglog}{polyloglog}
\let\Pr\relax\DeclareMathOperator*{\Pr}{\mathbb{Pr}}

\DeclareMathOperator{\minlab}{\mathsf{MinLab}}
\DeclareMathOperator{\maxpar}{\mathsf{MaxPar}}

\newcommand{\sss}{\mathsf{start}}
\newcommand{\ttt}{\mathsf{goal}}

\newcommand{\zo}{\{0,1\}}
\newcommand{\asg}{f}
\newcommand{\sqasg}{F}

\newcommand{\sqpi}{\Pi}

\newcommand{\SetCovReconf}{\prb{Set Cover Reconfiguration}}
\newcommand{\MinmaxSetCovReconf}{\prb{Minmax Set Cover Reconfiguration}}
\newcommand{\DomSetReconf}{\prb{Dominating Set Reconfiguration}}
\newcommand{\MinmaxDomSetReconf}{\prb{Minmax Dominating Set Reconfiguration}}
\newcommand{\VerCovReconf}{\prb{Hypergraph Vertex Cover Reconfiguration}}
\newcommand{\MinmaxVerCovReconf}{\prb{Minmax Hypergraph Vertex Cover Reconfiguration}}
\newcommand{\ParBCSPReconf}{\prb{Partial 2CSP Reconfiguration}}
\newcommand{\MaxParBCSP}{\prb{Max Partial 2CSP}}
\newcommand{\LabCovReconf}{\prb{Label Cover Reconfiguration}}

\newcommand{\calC}{\mathcal{C}}

\newcommand{\calE}{\mathcal{E}}
\newcommand{\calF}{\mathcal{F}}

\newcommand{\calU}{\mathcal{U}}
\newcommand{\calV}{\mathcal{V}}
\newcommand{\calW}{\mathcal{W}}

\newcommand{\bbN}{\mathbb{N}}

\newcommand{\bbR}{\mathbb{R}}

\newcommand{\scrC}{\mathscr{C}}

\newcommand{\scrS}{\mathscr{S}}

\let\Pr\relax\DeclareMathOperator*{\Pr}{\mathbb{P}}

\newtheorem{theorem}{Theorem}[section]
\newtheorem{proposition}[theorem]{Proposition}
\newtheorem{lemma}[theorem]{Lemma}
\newtheorem{claim}[theorem]{Claim}
\newtheorem{corollary}[theorem]{Corollary}
\newtheorem{observation}[theorem]{Observation}

\theoremstyle{definition}
\newtheorem{definition}[theorem]{Definition}
\newtheorem{problem}[theorem]{Problem}

\newenvironment{claim*}{\begin{claim}}{\end{claim}}

\numberwithin{equation}{section}

\title{Optimal PSPACE-hardness of \\ Approximating Set Cover Reconfiguration}
\author{
Shuichi Hirahara\thanks{National Institute of Informatics, Japan.
\href{mailto:s\_hirahara@nii.ac.jp}{\texttt{s\_hirahara@nii.ac.jp}}
}
\and
Naoto Ohsaka\thanks{CyberAgent, Inc., Tokyo, Japan. \href{mailto:ohsaka\_naoto@cyberagent.co.jp}{\texttt{ohsaka\_naoto@cyberagent.co.jp}}; \href{mailto:naoto.ohsaka@gmail.com}{\texttt{naoto.ohsaka@gmail.com}}
}
}
\date{\today}

\begin{document}

\maketitle
\thispagestyle{empty}
\begin{abstract}\noindent In the \prb{Minmax Set Cover Reconfiguration} problem,
given a set system $\mathcal{F}$ over a universe and
its two covers $\mathcal{C}^\mathsf{start}$ and $\mathcal{C}^\mathsf{goal}$ of size $k$,
we wish to transform $\mathcal{C}^\mathsf{start}$ into $\mathcal{C}^\mathsf{goal}$ by
repeatedly adding or removing a single set of $\mathcal{F}$
while covering the universe in any intermediate state.
Then, the objective is to minimize the maximize size of any intermediate cover during transformation.
We prove that
\prb{Minmax Set Cover Reconfiguration} and \prb{Minmax Dominating Set Reconfiguration} are
$\mathsf{PSPACE}$-hard to approximate within a factor of
$2-\frac{1}{\operatorname{polyloglog} N}$,
where $N$ is the size of the universe and the number of vertices in a graph, respectively,
improving upon
{Ohsaka} (SODA 2024)~\cite{ohsaka2024gap} and
{Karthik~C.~S.~and Manurangsi} (2023)~\cite{karthik2023inapproximability}.
This is the first result that exhibits a sharp threshold for the approximation factor of any reconfiguration problem 
because both problems admit a $2$-factor approximation algorithm as per
{Ito, Demaine, Harvey, Papadimitriou, Sideri, Uehara, and Uno} (Theor.~Comput.~Sci., 2011)~\cite{ito2011complexity}.
Our proof is based on a reconfiguration analogue of the FGLSS reduction \cite{feige1996interactive}
from \emph{Probabilistically Checkable Reconfiguration Proofs}
of {Hirahara and Ohsaka} (2024)~\cite{hirahara2024probabilistically}.
We also prove that
for any constant $\varepsilon \in (0,1)$,
\prb{Minmax Hypergraph Vertex Cover Reconfiguration} on $\operatorname{poly}(\varepsilon^{-1})$-uniform hypergraphs
is $\mathsf{PSPACE}$-hard to approximate within a factor of $2-\varepsilon$.
\end{abstract}
\clearpage
\tableofcontents
\clearpage

\section{Introduction}\label{sec:intro}

\subsection{Background}\label{subsec:intro:background}

In the field of reconfiguration,
we study the reachability and connectivity over the space of feasible solutions under an adjacency relation.
Given a \emph{source problem} that asks the existence of a feasible solution,
its \emph{reconfiguration problem} requires to decide if there exists a \emph{reconfiguration sequence},
namely, a step-by-step transformation between a pair of feasible solutions
while always preserving the feasibility of any intermediate solution.
One of the reconfiguration problems we study in this paper is \SetCovReconf \cite{ito2011complexity},
whose source problem is \prb{Set Cover}.
In the \SetCovReconf problem,
for a set system $\calF$ over a universe $\calU$ and
its two covers $\calC^\sss$ and $\calC^\ttt$ of size $k$,
we seek a reconfiguration sequence from $\calC^\sss$ to $\calC^\ttt$
consisting only of covers of size at most $k+1$,
each of which is obtained from the previous one by adding or removing a single set of $\calF$.
Countless reconfiguration problems have been defined from a variety of source problems,
including Boolean satisfiability, constraint satisfaction problems, and graph problems.
Studying reconfiguration problems may help elucidate the structure of the solution space of combinatorial problems \cite{gopalan2009connectivity}.

The computational complexity of reconfiguration problems has the following trend:
a reconfiguration problem is likely to be $\PSPACE$-complete
if its source problem is intractable (say, $\NP$-complete); e.g.,
\prb{Set Cover} \cite{ito2011complexity}, 
\prb{$3$SAT} \cite{gopalan2009connectivity}, and
\prb{Independent Set} \cite{hearn2005pspace,hearn2009games};
a source problem in $\P$ frequently induces a reconfiguration problem in $\P$; e.g.,
\prb{Spanning Tree} \cite{ito2011complexity} and
\prb{$2$SAT} \cite{gopalan2009connectivity}.
Some exception are however known; e.g., 
\prb{$3$Coloring} \cite{cereceda2011finding} and
\prb{Shortest Path} \cite{bonsma2013complexity}.
We refer the readers to the surveys by
\citet{nishimura2018introduction,heuvel13complexity} and
the Combinatorial Reconfiguration wiki \cite{hoang2023combinatorial}
for more algorithmic and hardness results of reconfiguration problems.

To overcome the computational hardness of a reconfiguration problem,
we consider its \emph{optimization version},
which affords to relax the feasibility of intermediate solutions.
For example,
\MinmaxSetCovReconf \cite{ito2011complexity} is an optimization version of \SetCovReconf,
where we are allowed to use any cover of size greater than $k+1$,
but required to minimize the \emph{maximum size of any covers} in the reconfiguration sequence
(see \cref{subsec:appl:SetCover} for the formal definition).
Solving this problem approximately,
we may be able to find a ``reasonable'' reconfiguration sequence for \SetCovReconf
which consists of covers of size at most, say, $1\%$ larger than $k+1$.
Unlike \prb{Set Cover}, which is $\NP$-hard to approximate within a factor smaller than $\ln n$ 
\cite{dinur2014analytical,feige1998threshold,lund1994hardness},
\MinmaxSetCovReconf admits a $2$-factor approximation algorithm due to \citet[Theorem~6]{ito2011complexity}.
An immediate question is: \emph{Is this the best possible?}

Here, we summarize known hardness-of-approximation results on \MinmaxSetCovReconf.
\citet{ohsaka2024gap} showed that
\MinmaxSetCovReconf is $\PSPACE$-hard to approximate within a factor of $1.0029$
assuming the Reconfiguration Inapproximability Hypothesis \cite{ohsaka2023gap}, which was recently proved \cite{hirahara2024probabilistically,karthik2023inapproximability}. 
\citet{karthik2023inapproximability} proved
$\NP$-hardness of the $(2-\epsilon)$-factor approximation for any constant $\epsilon \in (0,1)$.
Both results are not optimal:
\citeauthor{ohsaka2024gap}'s factor $1.0029$ is far smaller than $2$, while
\citeauthor{karthik2023inapproximability}'s result is not $\PSPACE$-hardness.
This leaves a tantalizing possibility that there may exist a polynomial-length reconfiguration sequence that achieves a $1.0030$-factor approximation for \MinmaxSetCovReconf,
and hence the approximation problem may be in $\NP$.
Note that the $\PSPACE$-hardness result of \citeauthor{ohsaka2024gap} disproves the existence of a polynomial-length witness  (in particular, a polynomial-length reconfiguration sequence) for the $1.0029$-factor approximation under the assumption that $\NP \neq \PSPACE$.

\subsection{Our Results}
\label{subsec:intro:results}

We present optimal results
of $\PSPACE$-hardness of approximation for three reconfiguration problems.
Our first result is that
\MinmaxSetCovReconf is $\PSPACE$-hard to approximate within a factor smaller than $2$,
improving upon \citet[Corollary~4.2]{ohsaka2024gap} and \citet[Theorem~4]{karthik2023inapproximability}.
This is the first result that exhibits a sharp threshold for the approximation factor of any reconfiguration problem:
approximating within any factor below $2$ is $\PSPACE$-complete and within a $2$-factor is in $\P$ \cite{ito2011complexity}.

\begin{theorem}[informal; see \cref{thm:SetCover}]
\label{intro:thm:SetCover}
For a set system $\calF$ of universe size $N$ and its two covers $\calC^\sss$ and $\calC^\ttt$ of size $k$,
it is $\PSPACE$-complete to distinguish between the following cases\textup{:}
\begin{itemize}
    \item \textup{(}Completeness\textup{)}
    There exists a reconfiguration sequence from
    $\calC^\sss$ to $\calC^\ttt$ consisting only of covers of size at most $k+1$.
    \item \textup{(}Soundness\textup{)}
    Every reconfiguration sequence contains a cover of size greater than
    $(2-\epsilon(N))(k+1)$,
    where $\epsilon(N) \coloneq (\polyloglog N)^{-1}$.
\end{itemize}
In particular,
\MinmaxSetCovReconf is $\PSPACE$-hard to approximate
within a factor of $2-\frac{1}{\polyloglog N}$.
\end{theorem}\noindent
As a corollary of \cref{thm:SetCover} along with \cite{ohsaka2024gap},
the following $\PSPACE$-hardness of approximation is established for \DomSetReconf,
which also admits a $2$-factor approximation \cite{ito2011complexity}
(please refer to \cite{ohsaka2024gap} for the problem definition).

\begin{corollary}[from \cref{thm:SetCover} and \protect{\cite[Corollary~4.3]{ohsaka2024gap}}]
\label{intro:cor:DominatingSet}
\MinmaxDomSetReconf is $\PSPACE$-hard to approximate within a factor of $2-\frac{1}{\polyloglog N}$,
where $N$ is the number of vertices in a graph.
\end{corollary}

Our third result is a similar inapproximability result for \VerCovReconf,
which is defined analogously to \SetCovReconf
(see
\ifthenelse{\boolean{FULL}}{\cref{subsec:appl:VertexCover}}{\cref{subsec:appl:VertexCover,app:appl}}
for the formal definition).
\MinmaxVerCovReconf is easily shown to be $2$-factor approximable \cite{ito2011complexity};
we prove that this is optimal.

\begin{theorem}[informal; see \cref{thm:VertexCover}]
\label{intro:thm:VertexCover}
For any constant $\epsilon \in (0,1)$,
a $\poly(\epsilon^{-1})$-uniform hypergraph, and
its two vertex covers $\calC^\sss$ and $\calC^\ttt$ of size $k$,
it is $\PSPACE$-complete to distinguish between the following cases\textup{:}
\begin{itemize}
    \item \textup{(}Completeness\textup{)}
    There exists a reconfiguration sequence from $\calC^\sss$ to $\calC^\ttt$
    consisting only of vertex covers of size at most $k+1$.
    \item \textup{(}Soundness\textup{)}
    Every reconfiguration sequence contains a vertex cover of size greater than $(2-\epsilon)(k+1)$.
\end{itemize}
In particular, \MinmaxVerCovReconf
on $\poly(\epsilon^{-1})$-uniform hypergraphs
is $\PSPACE$-hard to approximate within a factor of $2-\epsilon$.
\end{theorem}\noindent
We highlight here that 
the size of hyperedges in a \VerCovReconf instance of \cref{thm:VertexCover} depends (polynomially) only on the value of $\epsilon^{-1}$,
whereas the size of subsets in a \SetCovReconf instance of \cref{thm:SetCover} may depend on the universe size $N$.

\ifthenelse{\boolean{FULL}}{}{
Due to space limitation, proofs marked with $\ast$ are deferred to \cref{app:FGLSS,app:appl}.
}

\subsection{Proof Overview}
\label{subsec:intro:proof}

At a high level, our proofs of \cref{intro:thm:SetCover,intro:thm:VertexCover}
are given by combining the ideas developed in 
\cite{hirahara2024probabilistically,ohsaka2023gap,ohsaka2024gap,karthik2023inapproximability}.
\citet{karthik2023inapproximability}
proved $\NP$-hardness of the $(2-\epsilon)$-factor approximation of $\MinmaxSetCovReconf$ as follows.
\begin{enumerate}
    \item Starting from the PCP theorem for $\NP$ \cite{arora1998proof,arora1998probabilistic}, they applied the FGLSS reduction \cite{feige1996interactive} to prove $\NP$-hardness of the $\bigO(\epsilon^{-1})$-factor approximation of an intermediate problem, which we call \MaxParBCSP.
    \item The $\bigO(\epsilon^{-1})$-factor approximation of \MaxParBCSP is reduced to the $(2-\epsilon)$-factor approximation of a  reconfiguration problem, which we call \LabCovReconf (\cref{prb:MinLab}).
    \item \LabCovReconf can be reduced to \MinmaxSetCovReconf via approximation-preserving reductions  of \citet{ohsaka2024gap,lund1994hardness}.
\end{enumerate}
Here, \MaxParBCSP is defined as follows.
The input consists of a graph $G = (\calV, \calE)$, a finite alphabet $\Sigma$,
and constraints $\psi_e \colon \Sigma^2 \to \zo$ for each edge $e \in \calE$.
A \emph{partial assignment} is a function  $\asg \colon \calV \to \Sigma \cup \{\bot\}$, where the symbol $\bot$ indicates ``unassigned.''
The task is to maximize the fraction of assigned vertices in a partial assignment $\asg$ that \emph{satisfies} $\psi_e$ for every $e = (v, w) \in \calE$; i.e., $\psi_e(\asg(v), \asg(w)) = 1$ if $\asg(v) \neq \bot$ and $\asg(w) \neq \bot$.

To improve this $\NP$-hardness result to $\PSPACE$-hardness,
we replace the starting point with 
the PCRP (Probabilistically Checkable Reconfiguration Proof) system of \citet{hirahara2024probabilistically},
which is a reconfiguration analogue of the PCP theorem.
We also replace \MaxParBCSP with its reconfiguration analogue, which we call \ParBCSPReconf (\cref{prb:MaxPar}).
The proof of $\PSPACE$-hardness is outlined as follows.
\begin{enumerate}
    \item Starting from the \emph{PCRP theorem} for $\PSPACE$ \cite{hirahara2024probabilistically},
    we apply the FGLSS reduction \cite{feige1996interactive} to prove $\PSPACE$-hardness of \ParBCSPReconf (\cref{subsec:FGLSS:subconstant,subsec:FGLSS:MaxPar}).
    \item
    We reduce \ParBCSPReconf to \LabCovReconf (\cref{subsec:FGLSS:LabelCover}).
    \item
    We reduce \LabCovReconf to \MinmaxSetCovReconf by the reductions of \cite{ohsaka2024gap,lund1994hardness} (\cref{subsec:appl:SetCover}).
\end{enumerate}
The second and third steps are similar to the previous work \cite{karthik2023inapproximability}.
Our main technical contribution lies in the first step, which we explain below.

Roughly speaking, the PCRP theorem \cite{hirahara2024probabilistically} shows that
any $\PSPACE$ computation on inputs of length $n$ can be encoded into a sequence $\pi^{(1)}, \cdots, \pi^{(T)} \in \zo^{\poly(n)}$ of exponentially many proofs such that any adjacent pair of proofs $\pi^{(t)}$ and $\pi^{(t+1)}$ differs in at most one bit, and each proof $\pi^{(t)}$ can be probabilistically checked by reading $q(n)$ bits of the proof and using $r(n)$ random bits, where $q(n) = \bigO(1)$ and $r(n) = \bigO(\log n)$.
The FGLSS reduction \cite{feige1996interactive} transforms such a proof system into a graph $G = (\calV, \calE)$, an alphabet $\Sigma$, and constraints $(\psi_e)_{e \in \calE}$
such that each vertex $v \in \calV \coloneq \zo^{r(n)}$ corresponds to a coin flip sequence of a verifier, each value $\alpha \in \Sigma = \zo^{q(n)}$ corresponds to a local view of the verifier,
and the constraints $\psi_e$ check the consistency of two local views of the verifier.
This reduction works in the case of the PCP theorem and proves $\NP$-hardness of \MaxParBCSP \cite{karthik2023inapproximability}.
However, the reduction does not work in the case of the PCRP theorem:
We need to ensure that the reconfiguration sequence of proofs $\pi^{(1)}, \cdots, \pi^{(T)}$ is transformed into a sequence of partial assignments $\asg^{(1)}, \cdots, \asg^{(T)}$, each adjacent pair of which differs in at most one vertex.
The issue is that changing one bit in the original proof $\pi^{(t)}$ may result in changing the assignments of many vertices in a partial assignment $\asg^{(t)} \colon \calV \to \Sigma \cup \{\bot\}$.

To address this issue, we employ the ideas developed in \cite{ohsaka2023gap,ohsaka2024gap},
called the \emph{alphabet squaring trick},
and modify the FGLSS reduction as follows.
Given a verifier that reads $q(n)$ bits of a proof,
we define 
the alphabet as $\Sigma = \{0, 1, 01\}^{q(n)}$.
Intuitively, the symbol ``$01$'' means that we are taking $0$ and $1$ simultaneously.
This enables us to construct a reconfiguration sequence of 
partial assignments $\asg^{(1)}, \cdots, \asg^{(T)}$ from a reconfiguration sequence of proofs $\pi^{(1)}, \cdots, \pi^{(T)}$.
Details can be found in \Cref{subsec:FGLSS:MaxPar}.


\subsection{Related Work}
\label{subsec:intro:related}

\citet{ito2011complexity} showed that
optimization versions of \prb{SAT Reconfiguration} and \prb{Clique Reconfiguration}
are $\NP$-hard to approximate,
relying on $\NP$-hardness of approximating \prb{Max SAT}~\cite{hastad2001some} and \prb{Max Clique}~\cite{hastad1999clique}, respectively.
Note that their $\NP$-hardness results are not optimal
since \prb{SAT Reconfiguration} and \prb{Clique Reconfiguration} are $\PSPACE$-complete.
Toward $\PSPACE$-hardness of approximation for reconfiguration problems,
\citet{ohsaka2023gap} proposed the \emph{Reconfiguration Inapproximability Hypothesis} (RIH),
which postulates that a reconfiguration analogue of \prb{Constraint Satisfaction Problem}
is $\PSPACE$-hard to approximate, and demonstrated
$\PSPACE$-hardness of approximation for many popular reconfiguration problems,
including those of 
\prb{$3$SAT}, \prb{Independent Set}, \prb{Vertex Cover}, \prb{Clique}, \prb{Dominating Set}, and \prb{Set Cover}.
\citet{ohsaka2024gap} adapted \citeauthor{dinur2007pcp}'s gap amplification \cite{dinur2007pcp}
to demonstrate that under RIH,
optimization versions of \prb{$2$CSP Reconfiguration} and \SetCovReconf are 
$\PSPACE$-hard to approximate within a factor of $0.9942$ and $1.0029$, respectively.

Very recently, \citet{karthik2023inapproximability,hirahara2024probabilistically}
announced the proof of RIH independently,
implying that the above $\PSPACE$-hardness results hold \emph{unconditionally}.
\citet{karthik2023inapproximability} further proved that
(optimization versions of) \prb{$2$CSP Reconfiguration} and \SetCovReconf
are $\NP$-hard to approximate within a factor smaller than $2$,
which is quantitatively tight
because both problems are (nearly) $2$-factor approximable.
Our result partially resolves
an open question of \cite[Section~6]{karthik2023inapproximability}:
``\emph{Can we prove tight $\PSPACE$-hardness of approximation results for GapMaxMin-2-$\mathsf{CSP}_q$ and Set Cover Reconfiguration?}''

Other reconfiguration problems whose approximability was investigated include
those of 
\prb{Set Cover} \cite{ito2011complexity},
\prb{Subset Sum} \cite{ito2014approximability}, and
\prb{Submodular Maximization} \cite{ohsaka2022reconfiguration}.
We note that optimization variants of reconfiguration problems frequently
refer to those of \emph{the shortest reconfiguration sequence}
\cite{mouawad2017shortest,bonamy2020shortest,ito2022shortest,kaminski2011shortest},
which are orthogonal to this study.

\section{Preliminaries}
\label{sec:pre}

\subsection{Notations}
\label{subsec:pre:notations}

For a nonnegative integer $n \in \bbN$, let $ [n] \coloneq \{1, 2, \ldots, n\} $.
Unless otherwise specified, the base of logarithms is $2$.
A \emph{sequence} $\scrS$ of a finite number of objects $S^{(1)}, \ldots, S^{(T)}$
is denoted by $( S^{(1)}, \ldots, S^{(T)} )$, and
we write $S^{(t)} \in \scrS$ to indicate that $S^{(t)}$ appears in $\scrS$.
Let $\Sigma$ be a finite set called \emph{alphabet}.
For a length-$n$ string $\pi \in \Sigma^n$ and a finite sequence of indices $I \subseteq [n]^*$,
we use $\pi|_I \coloneq (\pi_i)_{i \in I}$ to denote the restriction of $\pi$ to $I$.
The \emph{Hamming distance} between two strings $f,g \in \Sigma^n$,
denoted by $\Delta(f,g)$,
is defined as the number of positions on which $f$ and $g$ differ; namely,
\begin{align}
    \Delta(f,g) \coloneq \left|\Bigl\{i \in [n] \Bigm| f_i \neq g_i \Bigr\}\right|.
\end{align}

\subsection{Reconfiguration Problems on Constraint Graphs}

\paragraph{Constraint Graphs.}
In this section, we formulate reconfiguration problems on constraint graphs.
The notion of \emph{constraint graph} is defined as follows.

\begin{definition}
A \emph{$q$-ary constraint graph} is defined as a tuple $G=(\calV,\calE,\Sigma,\Psi)$ such that
\begin{itemize}
    \item
    $(\calV,\calE)$ is a $q$-uniform\footnote{
        A hypergraph is said to be \emph{$q$-uniform} if each of its hyperedges has size exactly $q$.
    } hypergraph called the \emph{underlying graph},
    \item
    $\Sigma$ is a finite set called the \emph{alphabet}, and
    \item
    $\Psi = (\psi_e)_{e \in \calE}$ is a collection of $q$-ary \emph{constraints}, where
    each $\psi_e \colon \Sigma^e \to \zo$ is a circuit.
\end{itemize}
A binary constraint graph is simply referred to as a \emph{constraint graph}.
\end{definition}
For an \emph{assignment} $\asg \colon \calV \to \Sigma$,
we say that $\asg$ \emph{satisfies} a hyperedge $e = \{v_1, \ldots, v_q\} \in \calE$ (or a constraint $\psi_e$) if
$\psi_e(\asg(e)) = 1$,
where $\asg(e) \coloneq (\asg(v_1), \ldots, \asg(v_q))$, and
$\asg$ \emph{satisfies} $G$ if it satisfies all the hyperedges of $G$.
In the \prb{$q$CSP Reconfiguration} problem,
for a $q$-ary constraint graph $G$ and its two satisfying assignments $\asg^\sss$ and $\asg^\ttt$,
we are required to decide if there exists a reconfiguration sequence from $\asg^\sss$ to $\asg^\ttt$
consisting only of satisfying assignments for $G$,
each adjacent pair of which differs in at most one vertex.
\prb{$q$CSP Reconfiguration} is $\PSPACE$-complete in general \cite{gopalan2009connectivity,ito2011complexity};
thus, we formulate its two optimization versions.

\paragraph{\ParBCSPReconf.}
For a constraint graph $G = (\calV,\calE,\Sigma,\Psi = (\psi_e)_{e \in \calE})$,
a \emph{partial assignment} is defined as a function $\asg \colon \calV \to \Sigma \cup \{\bot\}$,
where the symbol $\bot$ indicates ``unassigned.''
We say that a partial assignment
$\asg \colon \calV \to \Sigma \cup \{\bot\}$
\emph{satisfies} an edge $e = (v,w) \in \calE$ if
$\psi_e(\asg(v), \asg(w)) = 1$ whenever $\asg(v) \neq \bot$ and $\asg(w) \neq \bot$.
The \emph{size} of $\asg$, denoted by $\|\asg\|$, is defined as
the number of vertices whose value is assigned; namely,
\begin{align}
    \|\asg\| \coloneq \left|\Bigl\{ v \in \calV \Bigm| \asg(v) \neq \bot \Bigr\}\right|.
\end{align}
For two satisfying partial assignments $\asg^\sss$ and $\asg^\ttt$ for $G$,
a \emph{reconfiguration partial assignment sequence from $\asg^\sss$ to $\asg^\ttt$}
is a sequence
$\sqasg = ( \asg^{(1)}, \ldots \asg^{(T)} )$
of satisfying partial assignments
such that
$\asg^{(1)} = \asg^\sss$,
$\asg^{(T)} = \asg^\ttt$, and
$\Delta(\asg^{(t)}, \asg^{(t+1)}) \leq 1$
(i.e., $\asg^{(t)}$ and $\asg^{(t+1)}$ differ in at most one vertex)
for all $t$.
For any reconfiguration partial assignment sequence
$\sqasg = ( \asg^{(1)}, \ldots, \asg^{(T)} )$,
we define $\|\sqasg\|_{\min}$ as
\begin{align}
    \|\sqasg\|_{\min} \coloneq \min_{1 \leq t \leq T} \|\asg^{(t)}\|.
\end{align}
\ParBCSPReconf is formally defined as follows:

\begin{problem}[\ParBCSPReconf]
\label{prb:MaxPar}
For a constraint graph
$G = (\calV,\calE,\Sigma,\Psi = (\psi_e)_{e \in \calE})$ and
its two satisfying partial assignments $\asg^\sss, \asg^\ttt \colon \calV \to \Sigma \cup \{\bot\}$,
we are required to find a reconfiguration partial assignment sequence $\sqasg$
from $\asg^\sss$ to $\asg^\ttt$
such that $\|\sqasg\|_{\min}$ is maximized.
\end{problem}
\noindent
Let $\maxpar_G(\asg^\sss \reco \asg^\ttt)$ denote the maximum value of
$\frac{\|\sqasg\|_{\min}}{|\calV|}$
over all possible reconfiguration sequences $\sqasg$ from $\asg^\sss$ to $\asg^\ttt$; namely,
\begin{align}
    \maxpar_G(\asg^\sss \reco \asg^\ttt)
    \coloneq \max_{\sqasg = ( \asg^\sss, \ldots, \asg^\ttt )} \frac{\|\sqasg\|_{\min}}{|\calV|}.
\end{align}
Note that $0 \leq \maxpar_G(\asg^\sss \reco \asg^\ttt) \leq 1$.
For every numbers $0 \leq s \leq c \leq 1$,
\prb{Gap$_{c,s}$ \ParBCSPReconf}
requests to determine for a constraint graph $G$ and its two satisfying partial assignments $\asg^\sss$ and $\asg^\ttt$,
whether $\maxpar_G(\asg^\sss \reco \asg^\ttt) \geq c$ or
$\maxpar_G(\asg^\sss \reco \asg^\ttt) < s$.
Note that we can assume
$\|\asg^\sss\| = \|\asg^\ttt\| = |\calV|$ when $c=1$.

\paragraph{\LabCovReconf.}
For a constraint graph $G = (\calV,\calE,\Sigma,\Psi)$,
a \emph{multi-assignment} is defined as a function
$\asg \colon \calV \to 2^\Sigma$.
We say that a multi-assignment $\asg$ \emph{satisfies} edge $e = (v,w) \in \calE$
if there exists a pair
$(\alpha, \beta) \in \asg(v) \times \asg(w)$
such that $\psi_e(\alpha, \beta) = 1$.
The size of $\asg$, denoted by $\|\asg\|$, is defined as
the sum of $|\asg(v)|$ over all $v \in \calV$; namely,
\begin{align}
    \|\asg\| \coloneq \sum_{v \in \calV} |\asg(v)|.
\end{align}
For two satisfying multi-assignments $\asg^\sss$ and $\asg^\ttt$ for $G$,
a \emph{reconfiguration multi-assignment sequence from $\asg^\sss$ to $\asg^\ttt$}
is a sequence $\sqasg = ( \asg^{(1)}, \ldots, \asg^{(T)} )$
of satisfying multi-assignments
such that
$\asg^{(1)} = \asg^\sss$,
$\asg^{(T)} = \asg^\ttt$, and
\begin{align}
    \sum_{v \in \calV} \Bigl| \asg^{(t)}(v) \triangle \asg^{(t+1)}(v) \Bigr| \leq 1 \; \text{ for all } t.
\end{align}
For any reconfiguration multi-assignment sequence
$\sqasg = ( \asg^{(1)}, \ldots, \asg^{(T)} )$,
we define $\|\sqasg\|_{\max}$ as
\begin{align}
    \|\sqasg\|_{\max} \coloneq \max_{1 \leq t \leq T} \|\asg^{(t)}\|.
\end{align}
\LabCovReconf is formally defined as follows.\footnote{
    This problem can be thought of as a reconfiguration analogue of \prb{Min Rep} \cite{charikar2011improved}.
}

\begin{problem}[\LabCovReconf]
\label{prb:MinLab}
For a constraint graph $G = (\calV,\calE,\Sigma,\Psi)$ and its two satisfying multi-assignments
$\asg^\sss, \asg^\ttt \colon \calV \to 2^\Sigma$,
we are required to find a reconfiguration multi-assignment sequence $\sqasg$
from $\asg^\sss$ to $\asg^\ttt$ such that
$\|\sqasg\|_{\max}$ is minimized.
\end{problem}
\noindent
Let $\minlab_G(\asg^\sss \reco \asg^\ttt)$
denote the minimum value of $\frac{\|\sqasg\|_{\max}}{|\calV|+1}$
over all possible reconfiguration multi-assignment sequences $F$ from $\asg^\sss$ to $\asg^\ttt$;
namely,
\begin{align}
    \minlab_G(\asg^\sss \reco \asg^\ttt)
    \coloneq \min_{\sqasg = ( \asg^\sss, \ldots, \asg^\ttt )}
    \frac{\|\sqasg\|_{\max}}{|\calV|+1}.
\end{align}
Note that $\minlab_G(\asg^\sss \reco \asg^\ttt) \geq \frac{|\calV|}{|\calV|+1}$.
For every numbers $1 \leq c \leq s$,
\prb{Gap$_{c,s}$ \LabCovReconf}
requests to determine
whether $\minlab_G(\asg^\sss \reco \asg^\ttt) \leq c$ or
$\minlab_G(\asg^\sss \reco \asg^\ttt) > s$
for a constraint graph $G$ and
its two satisfying multi-assignments $\asg^\sss$ and $\asg^\ttt$.
Note that we can assume
$\frac{\|\asg^\sss\|}{|\calV|+1} = \frac{\|\asg^\ttt\|}{|\calV|+1} \leq 1$ when $c=1$.

\subsection{Probabilistically Checkable Reconfiguration Proof Systems}
\label{subsec:pre:PCRP}

First, we formally define the notion of \emph{verifier}.

\begin{definition}
A \emph{verifier} with
\emph{randomness complexity} $r \colon \bbN \to \bbN$ and
\emph{query complexity} $q \colon \bbN \to \bbN$
is a probabilistic polynomial-time algorithm $V$
that given an input $x \in \zo^*$,
tosses $r = r(|x|)$ random bits $R$ and uses
$R$ to generate a sequence of $q = q(|x|)$ queries
$I = (i_1, \ldots, i_q)$ and a circuit $D \colon \zo^q \to \zo$.
We write $(I,D) \sim V(x)$ to denote the random variable
for a pair of the query sequence and circuit generated by $V$ on input $x \in \{0,1\}^*$.
Denote by $V^\pi(x) \coloneq D(\pi|_I)$ the output of $V$ on input $x$
given oracle access to a proof $\pi \in \{0,1\}^*$.
We say that $V(x)$ \emph{accepts} a proof $\pi$ if 
$V^\pi(x) = 1$; i.e.,
$D(\pi|_I) = 1$ for $(I,D) \sim V(x)$.
\end{definition}\noindent
We proceed to the definition of \emph{Probabilistically Checkable Reconfiguration Proofs} (PCRPs) due to 
\citet{hirahara2024probabilistically}, which offer
a PCP-type characterization of $\PSPACE$.
For any pair of proofs $\pi^\sss, \pi^\ttt \in \zo^\ell$,
a \emph{reconfiguration sequence from $\pi^\sss$ to $\pi^\ttt$}
is a sequence
$(\pi^{(1)}, \ldots, \pi^{(T)}) \in (\zo^\ell)^*$ such that
$\pi^{(1)} = \pi^\sss$,
$\pi^{(T)} = \pi^\ttt$, and 
$\Delta(\pi^{(t)}, \pi^{(t+1)}) \leq 1$
(i.e., $\pi^{(t)}$ and $\pi^{(t+1)}$ differ in at most one bit)
for all $t$.

\begin{theorem}[PCRP theorem of \protect{\citet{hirahara2024probabilistically}}]
\label{thm:PCRP}
For any language $L$ in \PSPACE,
there exists a verifier $V$ with
randomness complexity $r(n) = \bigO(\log n)$ and
query complexity $q(n) = \bigO(1)$,
coupled with polynomial-time computable functions
$\pi^\sss, \pi^\ttt \colon \zo^* \to \zo^*$, such that
the following hold for any input $x \in \zo^*$\textup{:}
\begin{itemize}
    \item \textup{(}Completeness\textup{)}
    If $x \in L$, there exists a reconfiguration sequence
    $\Pi = ( \pi^{(1)}, \ldots, \pi^{(T)} )$
    from $\pi^\sss(x)$ to $\pi^\ttt(x)$ over $\zo^{\poly(|x|)}$ such that
    $V(x)$ accepts every proof with probability $1$\textup{;} namely,
    \begin{align}
        \forall t \in [T], \quad \Pr\Bigl[V(x) \text{ accepts } \pi^{(t)}\Bigr] = 1.
    \end{align}
    \item \textup{(}Soundness\textup{)}
    If $x \notin L$, every reconfiguration sequence 
    $\Pi = ( \pi^{(1)}, \ldots, \pi^{(T)} )$
    from $\pi^\sss(x)$ to $\pi^\ttt(x)$ over $\zo^{\poly(|x|)}$
    includes a proof that is rejected by $V(x)$ with probability more than $\frac{1}{2}$\textup{;} namely,
    \begin{align}
        \exists t \in [T], \quad \Pr\Bigl[V(x) \text{ accepts } \pi^{(t)}\Bigr] < \frac{1}{2}.
    \end{align}
\end{itemize}    
\end{theorem}

We further introduce the notion of regular verifier.
We say that a verifier is \emph{regular} if each position in its proof is equally likely to be queried.\footnote{
    Note that regular verifiers are sometimes called \emph{smooth} verifiers, e.g., \cite{paradise2021smooth}.
    Since the term ``regularity'' is compatible with that of (hyper) graphs,
    we do not use the term ``smoothness'' but ``regularity.''
}

\begin{definition}
For a verifier $V$ and an input $x \in \zo^*$,
the \emph{degree} of a position $i$ of a proof is defined as
the number of times $i$ is queried by $V(x)$ over $r(|x|)$ random bits; namely,
\begin{align}
    \left|\Bigl\{ R \in \zo^{r(|x|)} \Bigm| i \in I_R \Bigr\}\right|
    = \Pr_{(I,D) \sim V(x)}\Bigl[i \in I\Bigr] \cdot 2^{r(|x|)},
\end{align}
where $r$ is the randomness complexity of $V$ and
$I_R$ is the query sequence generated by $V(x)$ on the randomness $R$.
A verifier $V$ is said to be \emph{$\Delta$-regular} if
the degree of every position is exactly equal to $\Delta$.
\end{definition}

\section{Subconstant Error PCRP Systems and FGLSS Reduction}
\label{sec:FGLSS}

In this section,
we construct a bounded-degree PCRP verifier with subconstant error
using \cref{thm:PCRP} in \cref{subsec:FGLSS:subconstant}, and
prove $\PSPACE$-hardness of approximation for 
\ParBCSPReconf and \LabCovReconf
by the FGLSS reduction \cite{feige1996interactive} in \cref{subsec:FGLSS:MaxPar,subsec:FGLSS:LabelCover}, respectively.

\subsection{Bounded-degree PCRP Systems with Subconstant Error}
\label{subsec:FGLSS:subconstant}

Starting from \cref{thm:PCRP},
we first obtain a regular PCRP verifier for any $\PSPACE$ language,
whose proof uses the degree reduction technique due to \citet{ohsaka2023gap}.

\ifthenelse{\boolean{FULL}}{
\begin{proposition}
}{
\begin{proposition}[$\ast$]
}
\label{prp:regular}
For any language $L$ in $\PSPACE$,
there exists a $\Delta$-regular PCRP verifier $V$ with
randomness complexity $r(n) = \bigO(\log n)$,
query complexity $q(n) = \bigO(1)$,
perfect completeness, and
soundness $1-\epsilon$, for some constant $\Delta \in \bbN$ and $\epsilon \in (0,1)$.
\end{proposition}
\ifthenelse{\boolean{FULL}}{
\begin{proof}
\begin{table}[t]
    \centering
    \small
    \begin{tabularx}{\textwidth}{c|cXXl}
        \toprule
        \textbf{problem/verifier} & \textbf{alph.~size} & \textbf{soundness} & \textbf{notes} & \textbf{reference} \\
        \midrule
        $q$-query verifier & $2$
            & $\frac{1}{2}$
            & ---
            & \cite[Theorem~5.1]{hirahara2024probabilistically} \\
        \midrule
        \prb{$2$CSP Reconf} & $3$
            & $1-\epsilon$ ($\epsilon$ depends on $q$)
            & ---
            & \cite[Lemmas~3.2~and~3.6]{ohsaka2023gap} \\
        \midrule
        \prb{$2$CSP Reconf} & $6$
            & $1-\bar{\epsilon}$ ($\bar{\epsilon}$ depends on $\epsilon$)
            & max.~degree $\Delta$ depends on $\epsilon$
            & \cite[Lemma~3.7]{ohsaka2023gap} \\
        \midrule
        \prb{$3$SAT Reconf} & $2$
            & $1-\Omega(\bar{\epsilon})$
            & each variable appears $\bigO(\Delta)$ times
            & \cite[Lemma~3.2]{ohsaka2023gap} \\
        \midrule
        $3$-query verifier & $2$
            & $1-\Omega\left(\frac{\bar{\epsilon}}{\Delta}\right)$
            & $\bigO(\Delta)$-regular
            & --- \\
        \bottomrule
    \end{tabularx}
    \caption{Sequence of reductions used in \cref{prp:regular}.}
    \label{tab:reductions}
\end{table}

Here, the suffix ``$_W$'' will designate the restricted case of \prb{$q$CSP Reconfiguration}
whose alphabet size $|\Sigma|$ is $W$.
By \cref{thm:PCRP}, for some integer $q \in \bbN$,
there is a PCRP verifier $V$ for $L$ with
randomness complexity $\bigO(\log n)$,
query complexity $q$,
perfect completeness, and
soundness $\frac{1}{2}$.
For an input $x \in \zo^*$,
the verifier $V(x)$ can be transformed into an instance of
\prb{Gap$_{1,\frac{1}{2}}$ $q$CSP$_2$ Reconfiguration} in
a canonical manner, e.g., \cite[Proposition~4.9]{hirahara2024probabilistically}.
By \cite[Lemmas~3.2~and~3.6]{ohsaka2023gap},
we then obtain an instance of
\prb{Gap$_{1,1-\epsilon}$ $2$CSP$_3$ Reconfiguration},
where $\epsilon \in (0,1)$ depends only on $q$.
Further applying the degree reduction step of \cite[Lemma~3.7]{ohsaka2023gap},
we get an instance of
\prb{Gap$_{1,1-\bar{\epsilon}}$ $2$CSP$_6$ Reconfiguration},
whose underlying graph has maximum degree bounded by $\Delta$,
where $\bar{\epsilon} \in (0,1)$ and $\Delta \in \bbN$
depend only on $\epsilon$.
Using \cite[Lemma~3.2]{ohsaka2023gap},
an instance of
\prb{Gap$_{1,1-\Omega(\bar{\epsilon})}$ $3$SAT Reconfiguration} is produced,
where each Boolean variable appears in at most $\bigO(\Delta)$ clauses.
By padding with trivial constraints,
this instance is transformed into an instance of
\prb{Gap$_{1,1-\Omega(\frac{\bar{\epsilon}}{\Delta})}$ $3$CSP$_2$ Reconfiguration},
whose underlying graph is $\bigO(\Delta)$-regular.
That is to say, there exists a $3$-query $\bigO(\Delta)$-regular PCRP verifier $\tilde{V}$ for $L$ with
randomness complexity $\bigO(\log n)$,
query complexity $\bigO(1)$,
perfect completeness, and
soundness $1-\epsilon'$,
where $\epsilon' = \Omega\left(\frac{\bar{\epsilon}}{\Delta}\right)$,
as desired.
See \cref{tab:reductions} for a sequence of reductions used to obtain $\tilde{V}$.

\end{proof}
}{
}

Subsequently,
using a randomness-efficient sampler over expander graphs
(e.g., \cite[Section~3]{hoory2006expander}),
we construct a \emph{bounded-degree} PCRP verifier
with \emph{subconstant error}.

\ifthenelse{\boolean{FULL}}{
\begin{proposition}
}{
\begin{proposition}[$\ast$]
}
\label{prp:subconstant}
For any language $L$ in $\PSPACE$ and 
any function $\delta \colon \bbN \to \bbR$ with
$\delta(n) = \Omega(n^{-1})$,
there exists a bounded-degree PCRP verifier $V$ with 
randomness complexity $r(n) = \bigO(\log \delta(n)^{-1} + \log n)$,
query complexity $q(n) = \bigO(\log \delta(n)^{-1})$,
perfect completeness, and
soundness $\delta(n)$.
Moreover, for any input $x \in \zo^*$,
the degree of any position is
$\poly(\delta(|x|)^{-1})$.
\end{proposition}

\paragraph{Verifier Description.}
Our PCRP verifier is described as follows.
By \cref{prp:regular},
let $V$ be a $\Delta$-regular PCRP verifier for a $\PSPACE$-complete language $L$ with
randomness complexity $r(n) = \bigO(\log n)$,
query complexity $q(n) = q \in \bbN$,
perfect completeness, and
soundness $1-\epsilon$,
where $\Delta \in \bbN$ and $\epsilon \in (0,1)$.
The proof length, denoted by $\ell(n)$, is polynomially bounded since $\ell(n) \leq q(n) 2^{r(n)} = \poly(n)$.
Hereafter, for any $r(n)$ random bit sequence $R$,
let $I_R$ and $D_R$ respectively denote the query sequence
and circuit generated by $V(x)$ on the randomness $R$.
Given a function $\delta \colon \bbN \to \bbR$ with $\delta(n) = \Omega(n^{-1})$,
we construct the following verifier $\tilde{V}$:

\begin{itembox}[l]{\textbf{Bounded-degree verifier $\tilde{V}$ with subconstant error.}}
\begin{algorithmic}[1]
    \item[\textbf{Input:}]
        a $\Delta$-regular verifier $V$ with soundness $1-\epsilon$,
        a function $\delta \colon \bbN \to \bbR$, and
        an input $x \in \zo^n$.
    \item[\textbf{Oracle access:}]
        a proof $\pi \in \zo^{\ell(n)}$.
    \State construct a $(d,\lambda)$-expander graph $X$ over vertex set $\zo^{r(n)}$
        with $\frac{\lambda}{d} < \frac{\epsilon}{4}$.
    \State let $\rho \coloneq \left\lceil \frac{2}{\epsilon} \ln \delta(n)^{-1} \right\rceil = \bigO(\log \delta(n)^{-1})$.
    \State uniformly sample a $(\rho-1)$-length random walk
        $\vec{R} = (R_1, \ldots, R_\rho)$ over $X$
        using $r(n) + \rho \cdot \log d$ random bits.
    \For{\textbf{each} $1 \leq k \leq \rho$}
        \State execute $V(x)$ on $R_k$ to generate a query sequence $I_{R_k} = (i_1, \ldots, i_q)$ and a circuit $D_{R_k} \colon \zo^q \to \zo$.
        \If{$D_{R_k}(\pi|_{I_{R_k}}) = 0$}
            \State declare \textsf{reject}.
        \EndIf
    \EndFor
    \State declare \textsf{accept}.
\end{algorithmic}
\end{itembox}

\paragraph{Correctness.}
The perfect completeness and soundness for a fixed proof $\pi \in \zo^{\ell(n)}$ are shown below,
whose proof relies on the property about random walks over expander graphs
due to \citet{alon1995derandomized}.
\ifthenelse{\boolean{FULL}}{
\begin{claim}
}{
\begin{claim}[$\ast$]
}
\label{clm:subconstant:acceptance}
If $V(x)$ accepts $\pi$ with probability $1$,
then $\tilde{V}(x)$ accepts $\pi$ with probability $1$.
If $V(x)$ accepts $\pi$ with probability less than $1-\epsilon$,
then $\tilde{V}(x)$ accepts $\pi$ with probability less than $\delta(n)$.
\end{claim}

\ifthenelse{\boolean{FULL}}{
To prove \cref{clm:subconstant:acceptance},
we refer to the following property about random walks over expander graphs.

\begin{lemma}[\protect{\citet{alon1995derandomized}}]
\label{lem:expander-walk}
    Let
    $X$ be a $(d,\lambda)$-expander graph,
    $S$ be any vertex set of $X$, and
    $\vec{R} \coloneq (R_1, \ldots, R_\rho)$
    be a $\rho$-tuple of random variables denoting 
    the vertices of a uniformly chosen $(\rho-1)$-length random walk over $X$.
    Then, it holds that
    \begin{align}
        \left(\frac{|S|}{|\calV(X)|} - 2\frac{\lambda}{d}\right)^\rho
        \leq \Pr_{\vec{R}}\Bigl[\forall k \in [\rho], \; R_k \in S \Bigr]
        \leq \left(\frac{|S|}{|\calV(X)|} + 2\frac{\lambda}{d}\right)^\rho.
    \end{align}
\end{lemma}

\begin{proof}[Proof of \cref{clm:subconstant:acceptance}]
Suppose first $V(x)$ accepts $\pi$ with probability $1$; then,
it holds that
\begin{align}
    \Pr\Bigl[\tilde{V}(x) \text{ accepts } \pi\Bigr]
    = \Pr_{\vec{R}}\Bigl[\forall k \in [\rho], \; D_{R_k}(\pi|_{I_{R_k}}) = 1\Bigr]
    = 1.
\end{align}
Suppose then $V(x)$ accepts $\pi$ with probability less than $1-\epsilon$.
Define 
\begin{align}
    S \coloneq \Bigl\{R \in \zo^{r(n)} \Bigm| D_R(\pi|_{I_R}) = 1 \Bigr\}.
\end{align}
Note that $\frac{|S|}{|\calV(X)|} < 1-\epsilon$.
Then, it holds that
\begin{align}
    \Pr\Bigl[\tilde{V}(x) \text{ accepts } \pi\Bigr]
    = \Pr_{\vec{R}}\Bigl[ \forall k \in [\rho], \; D_{R_k}(\pi|_{I_{R_k}}) = 1 \Bigr]
    = \Pr_{\vec{R}}\Bigl[ \forall k \in [\rho], \; R_k \in S \Bigr].
\end{align}
Applying \cref{lem:expander-walk}, we derive
\begin{align}
\begin{aligned}
    \Pr_{\vec{R}}\Bigl[\forall k \in [\rho], \; R_k \in S\Bigr]
    & \leq \left( \frac{|S|}{|\calV(X)|} + 2\frac{\lambda}{d} \right)^\rho \\
    & < \left(1-\frac{\epsilon}{2}\right)^\rho
        & \text{(by } \frac{\lambda}{d} < \frac{\epsilon}{4} \text{)} \\
    & \leq \exp\left(-\frac{\epsilon}{2} \rho\right) \\
    & \leq \delta(n),
        & \text{(by } \rho = \left\lceil\frac{2}{\epsilon}\ln \delta(n)^{-1}\right\rceil \text{)}
\end{aligned}
\end{align}
completing the proof.
\end{proof}

We are now ready to prove \cref{prp:subconstant}.
\begin{proof}[Proof of \cref{prp:subconstant}]
We first show the perfect completeness and soundness.
Suppose $x \in L$, then
there exists a reconfiguration sequence
$\sqpi = ( \pi^{(1)}, \ldots, \pi^{(T)} )$
from $\pi^\sss(x)$ to $\pi^\ttt(x)$ such that
$\Pr[V(x) \text{ accepts } \pi^{(t)}] = 1$ for all $t$.
By \cref{clm:subconstant:acceptance}, we have 
$\Pr[\tilde{V}(x) \text{ accepts } \pi^{(t)}] = 1$ for all $t$.
Suppose $x \notin L$, then
for every reconfiguration sequence
$\sqpi = ( \pi^{(1)}, \ldots, \pi^{(T)} )$
from $\pi^\sss(x)$ to $\pi^\ttt(x)$,
it holds that
$\Pr[V(x) \text{ accepts } \pi^{(t)}] < 1-\epsilon$ for some $t$.
By \cref{clm:subconstant:acceptance}, we have
$\Pr[\tilde{V}(x) \text{ accepts } \pi^{(t)}] < \delta(n)$ for such $t$.

Since $\rho = \bigO(\log \delta(n)^{-1})$,
the randomness complexity of $\tilde{V}$ is equal to
$\tilde{r}(n) = r(n) + \rho \cdot \log d = \bigO(\log \delta(n)^{-1} + \log n)$, and
the query complexity is
$\tilde{q}(n) = q(n) \cdot \rho = \bigO(\log \delta(n)^{-1})$.
Note that
$d$ and $\lambda$ may depend only on $\epsilon$, and
a $(d,\lambda)$-expander graph $X$ over $\zo^{r(n)}$ can be constructed in polynomial time in $2^{r(n)} = \poly(n)$, e.g.,
by using an explicit construction of near-Ramanujan graphs
\cite{mohanty2021explicit,alon2021explicit}.

Observe finally that $\tilde{V}$ queries each position $i \in [\ell(n)]$ of a proof with probability equal to
\begin{align}
\label{eq:subconstant:degree}
    \Pr_{\vec{R}}\left[ \bigvee_{1 \leq k \leq \rho} \Bigl(i \in I_{R_k}\Bigr) \right].
\end{align}
\item Since $V$ is $\Delta$-regular, it holds that
\begin{align}
    \Pr_{R \sim \zo^{r(n)}}\Bigl[i \in I_R \Bigr]
    = \frac{\Delta}{2^{r(n)}}.
\end{align}
Using the fact that each $R_k$ is uniformly distributed over $\zo^{r(n)}$,
we bound \cref{eq:subconstant:degree} as follows:
\begin{align}
    \Pr_{\vec{R}}\left[ \bigvee_{1 \leq k \leq \rho} \Bigl(i \in I_{R_k}\Bigr) \right]
    \underbrace{\leq}_{\text{union bound}} \sum_{k \in [\rho]}\Pr_{\vec{R}}\Bigl[i \in I_{R_k}\Bigr]
    = \frac{\rho \cdot \Delta}{2^{r(n)}}
    = \bigO\left(\frac{\log \delta(n)^{-1}}{2^{r(n)}}\right).
\end{align}
Consequently, the degree of each position $i$ with respect to $\tilde{V}$ is at most
\begin{align}
\begin{aligned}    
    \Pr_{\vec{R}}\left[ \bigvee_{1 \leq k \leq \rho} \Bigl(i \in I_{R_k}\Bigr) \right]
    \cdot 2^{\tilde{r}(n)}
    & = \bigO\left(\frac{\log \delta(n)^{-1}}{2^{r(n)}}\right)
    \cdot 2^{r(n) + \rho \cdot \log d} \\
    & = \bigO(\log \delta(n)^{-1}) \cdot 2^{\bigO(\log \delta(n)^{-1})} \\
    & = \poly(\delta(n)^{-1}),
\end{aligned}
\end{align}
which completes the proof.
\end{proof}

}{
We give a proof sketch of \cref{prp:subconstant}.
\begin{proof}[Proof sketch of \cref{prp:subconstant}]
The perfect completeness and soundness of $\tilde{V}$ are immediate from 
\cref{clm:subconstant:acceptance}.
By construction, it is easy to bound the randomness and query complexities of $\tilde{V}$.
Since $V$ is assumed to be $\Delta$-regular and
each random bit sequence $R_k$ is uniformly distributed over $\zo^{r(n)}$,
we can bound the degree of $\tilde{V}$ from above by $\poly(\delta(n)^{-1})$.
\end{proof}
}

\subsection{FGLSS Reduction and $\PSPACE$-hardness of Approximation for \ParBCSPReconf}
\label{subsec:FGLSS:MaxPar}

We now establish the FGLSS reduction from \cref{prp:subconstant} and
show that \ParBCSPReconf is $\PSPACE$-hard to approximate within a factor arbitrarily close to $0$.

\begin{theorem}
\label{thm:MaxPar}
For any function $\epsilon \colon \bbN \to \bbR$ with
$\epsilon(n) = \Omega\left(\frac{1}{\polylog n}\right)$,
\prb{Gap$_{1,\epsilon(N)}$ \ParBCSPReconf}
with alphabet size $\poly(\epsilon(N)^{-1})$ is \PSPACE-complete,
where $N$ is the number of vertices.
\end{theorem}

\paragraph{Reduction.}
We describe a reduction from a bounded-degree PCRP verifier to \ParBCSPReconf. 
Define $\delta(n) \coloneq \frac{\epsilon(\poly (n))}{2}$, whose precise expression is given later.
For any \PSPACE-complete language $L$,
let $V$ be a bounded-degree PCRP verifier of \cref{prp:subconstant} with
randomness complexity $r(n) = \bigO(\log \delta(n)^{-1} + \log n)$,
query complexity $q(n) = \bigO(\log \delta(n)^{-1})$,
perfect completeness, and
soundness $\delta(n)$.
The proof length $\ell(n)$ is polynomially bounded.
Suppose we are given an input $x \in \zo^n$.
Let $\pi^\sss, \pi^\ttt \in \zo^{\ell(n)}$ be the two proofs associated with $V(x)$.
Because the degree of $V$ is bounded by $\poly(\delta(n)^{-1})$,
for some constant $\kappa \in \bbN$,
we have
\begin{align}
    \Pr_{(I,D) \sim V(x)}\Bigl[i \in I\Bigr]
    \leq \frac{\delta(n)^{-\kappa}}{2^{r(n)}}
    \text{ for any } i \in [\ell(n)].
\end{align}
Hereafter, for any $r(n)$ random bit sequence $R$,
let $I_R$ and $D_R$ denote the query sequence
and the circuit generated by $V(x)$ on the randomness $R$, respectively.
Construct a constraint graph $G = (\calV,\calE,\Sigma,\Psi)$ as follows:
\begin{align}
    \calV & \coloneq \zo^{r(n)}, \\
    \calE & \coloneq
        \Bigl\{(R_1,R_2) \in \calV\times \calV \Bigm| I_{R_1} \cap I_{R_2} \neq \emptyset \Bigr\}, \\
    \Sigma & \coloneq
        \Bigl\{\{0\}, \{1\}, \{0,1\} \Bigr\}^{q(n)}, \\
    \Psi & \coloneq \{\psi_e\}_{e \in \calE},
\end{align}
where we define $\psi_{R_1,R_2} \colon \Sigma \times \Sigma \to \zo$
for each edge $(R_1,R_2) \in \calE$ so that
$\psi_{R_1,R_2}(\asg(R_1), \asg(R_2)) = 1$ for
an assignment $\asg \colon \calV \to \Sigma$
if and only if the following three conditions are satisfied:
\begin{align}
    & \forall \alpha \in \prod_{i \in I_R} \asg(R_1)_i, \quad D_{R_1}(\alpha) = 1,
    \label{eq:MaxPar:R1} \\
    & \forall \beta \in \prod_{i \in I_R} \asg(R_2)_i, \quad D_{R_2}(\beta) = 1,
    \label{eq:MaxPar:R2} \\
    & \forall i \in I_{R_1} \cap I_{R_2}, \quad \asg(R_1)_i \subseteq \asg(R_2)_i \text{ or } \asg(R_1)_i \supseteq \asg(R_2)_i.
    \label{eq:MaxPar:R1R2}
\end{align}
Here, for the sake of simple notation,
we consider $\asg(R)$ as if it were indexed by $I_R$ (rather than $[q(n)]$);
namely, $\asg(R) \in \{\{0\}, \{1\}, \{0,1\}\}^{I_R}$.
Thus, $\asg(R)$ for each $R \in \calV$ corresponds the local view of $V(x)$ on $R$.

For any proof $\pi \in \zo^{\ell(n)}$, we associate it with
an assignment $\asg_\pi \colon \calV \to \Sigma$ such that
\begin{align}
\label{eq:MaxPar:asgmtpi}
    \asg_\pi(R) \coloneq \Bigl(\{\pi_i\}\Bigr)_{i \in I_R} \text{ for all } R \in \calV.
\end{align}
Note that $\asg_\pi(R) \in \{\{0\}, \{1\}\}^{I_R}$.
Constructing two assignments
$\asg^\sss$ from $\pi^\sss$ and $\asg^\ttt$ from $\pi^\ttt$ by \cref{eq:MaxPar:asgmtpi},
we obtain an instance $(G,\asg^\sss,\asg^\ttt)$ of \ParBCSPReconf.
Observe that
$\asg^\sss$ and $\asg^\ttt$ satisfy $G$ and
$\|\asg^\sss\| = \|\asg^\ttt\| = |\calV|$.
Note that $N \coloneq |\calV| \leq n^c$ for some constant $c \in \bbN$.
Letting $\delta(n) \coloneq \frac{\epsilon(n^c)}{2} = \Omega\left(\frac{1}{\polylog n}\right)$
ensures that the alphabet size is $|\Sigma| = \bigO(3^{q(n)}) = \poly(\epsilon(N)^{-1})$.
This completes the description of the reduction.

\paragraph{Correctness.}
We first prove the completeness.
\begin{lemma}[Completeness]
\label{lem:MaxPar:completeness}
If $x \in L$, then $\maxpar_G(\asg^\sss \reco \asg^\ttt) = 1$.
\end{lemma}

\begin{proof}[Proof of \cref{lem:MaxPar:completeness}]
It is sufficient to consider the case that
$\pi^\sss$ and $\pi^\ttt$ differ in exactly one position, say, $i^\star \in [\ell(n)]$; namely,
$\pi^\sss_{i^\star} \neq \pi^\ttt_{i^\star}$ and
$\pi^\sss_i = \pi^\ttt_i$ for all $i \neq i^\star$.
Note that $\asg^\sss$ and $\asg^\ttt$ may differ in two or more vertices.
Consider a reconfiguration partial assignment sequence $\sqasg$ from $\asg^\sss$ to $\asg^\ttt$
obtained by the following procedure:
\begin{itembox}[l]{\textbf{Reconfiguration sequence $\sqasg$ from $\asg^\sss$ to $\asg^\ttt$.}}
\begin{algorithmic}[1]
    \For{\textbf{each} $R \in \calV$ such that $i^\star \in I_R$}
        \State change the \nth{$i^\star$} entry of $R$'s current value
        from $\asg^\sss(R)_{i^\star} = \{\pi^\sss_{i^\star}\}$
        to $\zo$.
    \EndFor
    \For{\textbf{each} $R \in \calV$ such that $i^\star \in I_R$}
        \State change the \nth{$i^\star$} entry of $R$'s current value
        from $\zo$
        to $\asg^\ttt(R)_{i^\star} = \{\pi^\ttt_{i^\star}\}$.
    \EndFor
\end{algorithmic}
\end{itembox}
Observe that any partial assignment $\asg^\circ$ of $\sqasg$ satisfies $G$ for the following reasons:
\begin{itemize}
\item
Since $f^\circ(R)_{i^\star} \subseteq \{0,1\} = \{\pi^\sss_{i^\star}, \pi^\ttt_{i^\star}\} = f^\sss(R)_{i^\star} \cup f^\ttt(R)_{i^\star}$
when $i^\star \in I_R$,
$\asg^\circ$ satisfies \cref{eq:MaxPar:R1,eq:MaxPar:R2}.
\item Letting $K \coloneq \{ f^\circ(R)_{i^\star} \mid i^\star \in I_R \}$,
we find $K$ to be either
$\{\{0\}\}$,
$\{\{1\}\}$,
$\{\{0,1\}\}$,
$\{\{0\},\{0,1\}\}$, or
$\{\{1\},\{0,1\}\}$ by construction;
i.e., $\asg^\circ$ satisfies \cref{eq:MaxPar:R1R2}.
\end{itemize}
Since $\|\asg^\circ\| = |\calV|$,
it holds that
$\maxpar_G(\asg^\sss \reco \asg^\ttt) \geq \frac{\|\sqasg\|_{\max}}{|\calV|} = 1$,
completing the proof.
\end{proof}

\begin{lemma}[Soundness]
\label{lem:MaxPar:soundness}
If $x\notin L$, then
\begin{align}
    \maxpar_G(\asg^\sss \reco \asg^\ttt)
    < \delta(n) + \frac{q(n) \cdot \delta(n)^{-\kappa}}{2^{r(n)}}.
\end{align}
\end{lemma}

The proof of \cref{thm:MaxPar} follows from 
\cref{lem:MaxPar:completeness,lem:MaxPar:soundness} because
for any sufficiently large $n$ such that
$\frac{q(n) \cdot \delta(n)^{-\kappa}}{2^{r(n)}} \leq \delta(n)$
(note that $\delta(n) = \Omega\left(\frac{1}{\polylog n}\right)$),
the following hold:
\begin{itemize}
    \item (Perfect completeness) If $x \in L$,
        then $\maxpar_G(\asg^\sss \reco \asg^\ttt) = 1$;
    \item (Soundness) If $x \notin L$,
        then $\maxpar_G(\asg^\sss \reco \asg^\ttt) < 2\delta(n) = \epsilon(N)$.
\end{itemize}

\begin{proof}[Proof of \cref{lem:MaxPar:soundness}]
We prove the contrapositive.
Suppose
$\maxpar_G(\asg^\sss \reco \asg^\ttt) \geq \Gamma$ for some $\Gamma \in (0,1)$, and
there is a reconfiguration partial assignment sequence
$\sqasg = ( \asg^{(1)}, \ldots, \asg^{(T)} )$
from $\asg^\sss$ to $\asg^\ttt$
such that
$\|\sqasg\|_{\min} = \maxpar_G(\asg^\sss \reco \asg^\ttt)$.
Define then a (not necessarily reconfiguration) sequence
$\sqpi = ( \pi^{(1)}, \ldots, \pi^{(T)} )$
over $\zo^{\ell(n)}$ such that
each proof $\pi^{(t)}$
is determined based on the \emph{plurality vote} over $\asg^{(t)}$; namely,
\begin{align}
    \pi^{(t)}_i \coloneq \argmax_{b \in \zo}
    \left|\Bigl\{ R \in \calV \Bigm| i \in I_R \text{ and } b \in \asg^{(t)}(R)_i \Bigr\}\right|
    \; \text{for all } i \in [\ell(n)],
\end{align}
where ties are broken so that $0$ is chosen.
In particular,
$\pi^{(1)} = \pi^\sss$ and $\pi^{(T)} = \pi^\ttt$.
Observe the following:

\begin{observation}
\label{obs:MaxPar:popularity}
    For any $t \in [T]$ and $R \in \calV$, it holds that
    \begin{align}\label{eq:MaxPar:popularity}
        \asg^{(t)}(R) \neq \bot
        \implies D_R(\pi^{(t)}|_{I_R}) = 1.
    \end{align}
\end{observation}

Since
$\Pr_{R \sim \calV} [f^{(t)}(R) \neq \bot] =  \|\asg^{(t)}\| \geq \Gamma$,
by \cref{obs:MaxPar:popularity},
we have that for all $t$,
\begin{align}
    \Pr\Bigl[ V(x) \text{ accepts } \pi^{(t)} \Bigr]
    = \Pr_{R \sim \zo^{r(n)}}\Bigl[ D_R(\pi^{(t)}|_{I_R}) = 1 \Bigr]
    \geq \Pr_{R \sim \calV} \Bigl[f^{(t)}(R) \neq \bot\Bigr]
    \geq \Gamma.
\end{align}
Unfortunately, $\sqpi$ is \emph{not} a reconfiguration sequence because
$\pi^{(t)}$ and $\pi^{(t+1)}$ may differ in two or more positions.
Since $f^{(t)}$ and $f^{(t+1)}$ differ in a single vertex $R \in \calV$,
we have $\pi^{(t)}_i \neq \pi^{(t+1)}_i$ only if $i \in I_R$, implying
$\Delta(\pi^{(t)}, \pi^{(t+1)}) \leq |I_R| = q(n)$.
Using this fact,
we interpolate between $\pi^{(t)}$ and $\pi^{(t+1)}$
to find a valid reconfiguration sequence $\sqpi^{(t)}$ such that
$V(x)$ accepts every proof of $\sqpi^{(t)}$ with probability $\Gamma - o(1)$.

\begin{claim}
\label{clm:MaxPar:interpolate}
    There exists a reconfiguration sequence $\sqpi^{(t)}$
    from $\pi^{(t)}$ to $\pi^{(t+1)}$
    such that for every proof $\pi^\circ$ of $\sqpi^{(t)}$,
    \begin{align}
        \Pr\Bigl[V(x) \text{ accepts } \pi^\circ\Bigr] 
        \geq \Gamma - \frac{q(n) \cdot \delta(n)^{-\kappa}}{2^{r(n)}}.
    \end{align}
\end{claim}

Concatenating $\sqpi^{(t)}$'s of \cref{clm:MaxPar:interpolate} for all $t$,
we obtain a valid reconfiguration sequence $\sqpi$ from
$\pi^\sss$ to $\pi^\ttt$ such that
\begin{align}
    \min_{1 \leq t \leq T} \Pr\Bigl[V(x) \text{ accepts } \pi^{(t)} \Bigr]
    \geq \Gamma - \frac{q(n) \cdot \delta(n)^{-\kappa}}{2^{r(n)}}.
\end{align}
Substituting
$\delta(n) + \frac{q(n) \cdot \delta(n)^{-\kappa}}{2^{r(n)}}$
for $\Gamma$, we have that if
$\maxpar_G(\asg^\sss \reco \asg^\ttt) \geq \delta(n) + \frac{q(n)\cdot \delta(n)^{-\kappa}}{2^{r(n)}}$,
then $V(x)$ accepts every proof $\pi^{(t)}$ of $\sqpi$ with probability at least $\delta(n)$; i.e.,
$x \in L$.
This completes the proof of \cref{lem:MaxPar:soundness}.
\end{proof}

What remains to be done is to prove \cref{obs:MaxPar:popularity,clm:MaxPar:interpolate}.

\begin{proof}[Proof of \cref{obs:MaxPar:popularity}]
Suppose $\asg^{(t)}(R) \neq \bot$ for some $t \in [T]$ and $R \in \calV$.
We will show that $\pi^{(t)}_i \in \asg^{(t)}(R)_i$ for every $i \in I_R$.
Define
\begin{align}
    K \coloneq \Bigl\{
        \asg^{(t)}(R')_i \Bigm| \exists R' \in \calV \text{ s.t.~} i \in I_{R'} \text{ and } \asg^{(t)}(R') \neq \bot
    \Bigr\}.
\end{align}
Then, 
any pair $\alpha, \beta \in K$ must satisfy that
$\alpha \subseteq \beta$ or $\alpha \supseteq \beta$ because otherwise,
$\asg^{(t)}$ would violate \cref{eq:MaxPar:R1R2} at edge $(R_1,R_2)$ such that
$i \in R_1 \cap R_2$, $\asg^{(t)}(R_1)_i = \alpha$, and $\asg^{(t)}(R_2)_i = \beta$, which is a 
contradiction.
For each possible case of $K$,
the result of the plurality vote $\pi^{(t)}_i$ is shown below,
implying that $\pi^{(t)}_i \in \asg^{(t)}(R)_i$.

\begin{tabular}{c|cccccc}
\toprule
$K$ & $\{\}$ & $\{\{0\}\}$ & $\{\{1\}\}$ & $\{\{0,1\}\}$ & $\{\{0\}, \{0,1\}\}$ & $\{\{1\}, \{0,1\}\}$ \\
\midrule
$\pi^{(t)}_i$ & $0$ & $0$ & $1$ & $0$ & $0$ & $1$ \\
\bottomrule
\end{tabular}

Since $\asg^{(t)}(R)$ must satisfy a self-loop $(R,R) \in \calE$,
by the definition of $\psi_{R,R}$, we have
\begin{align}
    \forall \alpha \in \prod_{i \in I_R} \asg^{(t)}(R)_i, \quad D_R(\alpha) = 1,
\end{align}
On the other hand, it holds that
\begin{align}
    \pi^{(t)}|_{I_R} \in \prod_{i \in I_R} \asg^{(t)}(R)_i,
\end{align}
implying $D_R(\pi^{(t)}|_{I_R}) = 1$, as desired.
\end{proof}

\begin{proof}[Proof of \cref{clm:MaxPar:interpolate}]
Recall that $\pi^{(t)}$ and $\pi^{(t+1)}$ may differ in at most $q(n)$ positions.
Consider any trivial reconfiguration sequence $\sqpi^{(t)}$
from $\pi^{(t)}$ to $\pi^{(t+1)}$
by simply changing at most $q(n)$ positions on which $\pi^{(t)}$ and $\pi^{(t+1)}$ differ.
By construction,
any proof $\pi^\circ$ of $\sqpi^{(t)}$ differs from $\pi^{(t)}$ in at most $q(n)$ positions, say,
$I^\circ \in {\ell(n) \choose \leq q(n)}$.
Then, we derive the following:
\begin{align*}
\begin{aligned}
    \Pr\Bigl[V(x) \text{ accepts } \pi^\circ\Bigr]
    & = \Pr_{(I,D) \sim V(x)}\Bigl[ D(\pi^\circ|_I) = 1\Bigr]
    \geq \Pr_{(I,D) \sim V(x)}\Bigl[ D(\pi^\circ|_I) = 1 \text{ and } I \cap I^\circ = \emptyset \Bigr] \\
    & = \Pr_{(I,D) \sim V(x)}\Bigl[ D(\pi^{(t)}|_I) = 1 \text{ and } I \cap I^\circ = \emptyset \Bigr] \\
    & = \underbrace{\Pr_{(I,D) \sim V(x)}\Bigl[ D(\pi^{(t)}|_I) = 1\Bigr]}_{= \Pr\left[V(x) \text{ accepts } \pi^{(t)}\right] \geq \Gamma}
        - \Pr_{(I,D) \sim V(x)}\Bigl[ D(\pi^{(t)}|_I) = 1 \text{ and } I \cap I^\circ \neq \emptyset \Bigr] \\
    & \geq \Gamma
        - \Pr_{(I,D) \sim V(x)}\Bigl[I \cap I^\circ \neq \emptyset \Bigr].
\end{aligned}
\end{align*}
Recall that
$\Pr_{(I,D) \sim V(x)}[i \in I] \leq \frac{\delta(n)^{-\kappa}}{2^{r(n)}}$
for any $i \in [\ell(n)]$ by assumption.
Since $|I^\circ| \leq q(n)$, taking a union bound, we have
\begin{align}
    \Pr_{(I,D) \sim V(x)}\Bigl[I \cap I^\circ \neq \emptyset \Bigr]
    \leq \sum_{i \in I^\circ}\Pr_{(I,D) \sim V(x)}\Bigl[i \in I\Bigr]
    \leq \frac{q(n) \cdot \delta(n)^{-\kappa}}{2^{r(n)}},
\end{align}
implying that
\begin{align}
    \Pr\Bigl[V(x) \text{ accepts } \pi^\circ\Bigr]
    \geq \Gamma - \frac{q(n) \cdot \delta(n)^{-\kappa}}{2^{r(n)}}.
\end{align}
This completes the proof.
\end{proof}

\subsection{Reducing \ParBCSPReconf to \LabCovReconf}
\label{subsec:FGLSS:LabelCover}

Subsequently, we show $\PSPACE$-hardness of approximation for \LabCovReconf
by reducing from \ParBCSPReconf,
whose proof is similar to \cite{karthik2023inapproximability}.
Note that \LabCovReconf admits a $2$-factor approximation,
similarly to \MinmaxSetCovReconf (see \cref{subsec:appl:SetCover}).

\ifthenelse{\boolean{FULL}}{
\begin{theorem}
}{
\begin{theorem}[$\ast$]
}
\label{thm:LabelCover}
For any function $\epsilon \colon \bbN \to \bbR$ with
$\epsilon(n) = \Omega\left(\frac{1}{\polylog n}\right)$,
\prb{Gap$_{1,2-\epsilon(N)}$ \LabCovReconf}
with alphabet size $\poly(\epsilon(N)^{-1})$
is $\PSPACE$-complete,
where $N$ is the number of vertices.
In particular, 
\begin{itemize}
\item for any constant $\epsilon \in (0,1)$,
\prb{Gap$_{1,2-\epsilon}$ \LabCovReconf}
with constant alphabet size is $\PSPACE$-complete, and
\item 
\prb{Gap$_{1,2-\frac{1}{\polyloglog N}}$ \LabCovReconf}
with alphabet size $\polyloglog N$ is $\PSPACE$-complete.
\end{itemize}
\end{theorem}
\ifthenelse{\boolean{FULL}}{
\begin{proof}
We present a gap-preserving reduction from
\prb{Gap$_{1,\frac{\epsilon(N)}{2}}$ \ParBCSPReconf}
to
\prb{Gap$_{1,2-\epsilon(N)}$ \LabCovReconf}.
Let
$(G=(\calV,\calE,\Sigma,\Psi),\asg^\sss,\asg^\ttt)$
be an instance of \ParBCSPReconf with $N$ variables and alphabet size $\poly(\epsilon(N)^{-1})$,
where $\|\asg^\sss\| = \|\asg^\ttt\| = |\calV|$.
Without loss of generality, we can safely assume that
$N$ is sufficiently large so that $N \geq \frac{4}{\epsilon(N)}$ because
$\epsilon(N) = \Omega\left(\frac{1}{\polylog N}\right)$.
Construct then multi-assignments
$\asg'^\sss, \asg'^\ttt \colon \calV \to 2^\Sigma$ such that
$\asg'^\sss(v) \coloneq \{\asg^\sss(v)\}$ and
$\asg'^\ttt(v) \coloneq \{\asg^\ttt(v)\}$
for all $v \in \calV$.
Observe that $\asg'^\sss$ and $\asg'^\ttt$ satisfy $G$;
thus, $(G,\asg'^\sss,\asg'^\ttt)$
is an instance of \LabCovReconf, completing the reduction.

We first show the perfect completeness; namely,
\begin{align}
    \maxpar_G(\asg^\sss \reco \asg^\ttt) \geq 1
    \implies
    \minlab_G(\asg'^\sss \reco \asg'^\ttt) \leq 1.
\end{align}
Suppose there is a reconfiguration partial assignment sequence
$\sqasg = ( \asg^{(1)}, \ldots, \asg^{(T)} )$ from
$\asg^\sss$ to $\asg^\ttt$ such that
$\|\sqasg\|_{\min} = |\calV|$.
Construct then a sequence
$\sqasg' = ( \asg'^{(1)}, \asg'^{(1.5)}, \ldots, \asg'^{(T-0.5)}, \asg'^{(T)} )$
of multi-assignments
such that
$\asg'^{(t)}(v) = \{\asg^{(t)}(v)\}$ for all $t \in [T]$ and $v \in \calV$, and
$\asg'^{(t+0.5)}$ for each $t \in [T-1]$ is defined as follows:
Given that $\asg^{(t)}$ and $\asg^{(t+1)}$ differ only in $v^\star$,
we let
\begin{align}
    \asg'^{(t+0.5)}(v) \coloneq
    \begin{cases}
        \left\{\asg^{(t)}(v^\star), \asg^{(t+1)}(v^\star)\right\} & \text{if } v = v^\star, \\
        \left\{\asg^{(t)}(v)\right\} & \text{otherwise},
    \end{cases}
    \text{ for all } v \in \calV.
\end{align}
In particular, it holds that
$\asg'^{(1)} = \asg'^\sss$ and
$\asg'^{(T)} = \asg'^\ttt$.
Observe that
$\asg'^{(t+0.5)}$ is a satisfying multi-assignment
with $\|\asg'^{(t+0.5)}\| = N+1$
for all $t$, and that
$\sum_{v \in \calV}|\asg'^{(t)} \triangle \asg'^{(t+0.5)}| = 1$; i.e.,
$\sqasg'$ is a reconfiguration multi-assignment sequence
from $\asg'^\sss$ to $\asg'^\ttt$
such that
$\|\sqasg'\|_{\max} = N+1$; i.e.,
$\minlab_G(\asg'^\sss \reco \asg'^\ttt) \leq \frac{\|\sqasg'\|_{\max}}{N+1} = 1$.

We then prove the soundness; i.e.,
\begin{align}
    \maxpar_G(\asg^\sss \reco \asg^\ttt) < \frac{\epsilon(N)}{2}
    \implies
    \minlab_G(\asg'^\sss \reco \asg'^\ttt) > 2-\epsilon(N).
\end{align}
Suppose we are given a reconfiguration multi-assignment sequence
$\sqasg' = ( \asg'^{(1)}, \ldots, \asg'^{(T)} )$
from $\asg'^\sss$ to $\asg'^\ttt$
such that
$\frac{\|\sqasg'\|_{\max}}{N+1} = \minlab_G(\asg'^\sss \reco \asg'^\ttt)$.
Define
\begin{align}
    \calV_1^{(t)} \coloneq \Bigl\{v \in \calV \Bigm| |\asg'^{(t)}(v)| = 1 \Bigr\}.
\end{align}
Construct then a sequence
$\sqasg = ( \asg^{(1)}, \ldots, \asg^{(T)} )$
of partial assignments
such that
each $\asg^{(t)} \colon \calV \to \Sigma \cup \{\bot\}$ is defined as follows:
\begin{align}
    \asg^{(t)}(v) \coloneq
    \begin{cases}
        \text{unique } \alpha \in \asg'^{(t)}(v) & \text{if } v \in \calV_1^{(t)}, \\
        \bot & \text{otherwise},
    \end{cases}
    \text{ for all } v \in \calV.
\end{align}
In particular, it holds that $\asg^{(1)} = \asg^\sss$ and $\asg^{(T)} = \asg^\ttt$.
Observe easily that
$\asg^{(t)}$ is a satisfying partial assignment, and
$\asg^{(t)}$ and $\asg^{(t+1)}$ differ in at most one vertex; i.e.,
$\sqasg$ is a reconfiguration partial assignment sequence
from $\asg^\sss$ to $\asg^\ttt$.
By assumption,
there exists a partial assignment $\asg^{(t)}$ in $\sqasg$ such that
$|\calV_1^{(t)}| = \|\asg^{(t)}\| < \frac{\epsilon(N)}{2} \cdot N$.
Simple calculation derives that
\begin{align}
\begin{aligned}
    \|\asg'^{(t)}\|
    & \geq 1\cdot |\calV_1^{(t)}| + 2\cdot |\calV \setminus \calV_1^{(t)}| \\
    & = 2N - |\calV_1^{(t)}| \\
    & > \left(2-\frac{\epsilon(N)}{2}\right) \cdot N \\
    & \underbrace{\geq}_{N \geq \frac{4}{\epsilon(N)}} (2-\epsilon(N)) \cdot (N+1),
\end{aligned}
\end{align}
implying that
$\minlab_G(\asg'^\sss \reco \asg'^\ttt)
= \frac{\|F'\|_{\max}}{N+1}
\geq \frac{\|\asg'^{(t)}\|}{N+1}
> 2-\epsilon(N)$,
which completes the proof.

\end{proof}
}{
}

\section{Applications}
\label{sec:appl}

In this section, we apply \cref{thm:LabelCover}
to show optimal $\PSPACE$-hardness of approximation results for
\MinmaxSetCovReconf (\cref{thm:SetCover}) and \MinmaxVerCovReconf (\cref{thm:VertexCover}).

\subsection{$\PSPACE$-hardness of Approximation for \SetCovReconf}
\label{subsec:appl:SetCover}

We first prove that \MinmaxSetCovReconf
is $\PSPACE$-hard to approximate within a factor smaller than $2$.
Let $\calU$ be a finite set called the \emph{universe} and
$\calF = \{ S_1, \ldots, S_m \}$ be a family of $m$ subsets of $\calU$.
A \emph{cover} for a set system $(\calU,\calF)$ is
a subfamily of $\calF$ whose union is equal to $\calU$.
For any pair of covers $\calC^\sss$ and $\calC^\ttt$ for $(\calU,\calF)$,
a \emph{reconfiguration sequence from $\calC^\sss$ to $\calC^\ttt$} is a sequence
$\scrC = (\calC^{(1)}, \ldots, \calC^{(T)})$ of covers such that
$\calC^{(1)} = \calC^\sss$,
$\calC^{(T)} = \calC^\ttt$, and
$|\calC^{(t)} \triangle \calC^{(t+1)}| \leq 1$
(i.e., $\calC^{(t+1)}$ is obtained from $\calC^{(t)}$ by adding or removing a single set of $\calF$)
for all $t$.
In \SetCovReconf \cite{ito2011complexity},
for a set system $(\calU,\calF)$ and its two covers $\calC^\sss$ and $\calC^\ttt$ of size $k$,
we are asked to decide if
there is a reconfiguration sequence from $\calC^\sss$ to $\calC^\ttt$
consisting only of covers of size at most $k+1$.
Next, we formulate its optimization version.
Denote by $\opt(\calF)$ the size of the minimum cover of $(\calU,\calF)$.
For any reconfiguration sequence
$\scrC = ( \calC^{(1)}, \ldots, \calC^{(T)} )$,
its \emph{cost} is defined as the maximum value of 
$\frac{|\calC^{(t)}|}{\opt(\calF)+1}$ over all $\calC^{(t)}$'s in $\scrC$; namely,\footnote{
    Here, division by $\opt(\calF)+1$ is derived from the nature that
    we must first add at least one set whenever
    $|\calC^\sss| = |\calC^\ttt| = \opt(\calF)$ and $\calC^\sss \neq \calC^\ttt$.
}
\begin{align}
    \cost_{\calF}(\scrC) \coloneq
    \max_{\calC^{(t)} \in \scrC} \frac{|\calC^{(t)}|}{\opt(\calF)+1},
\end{align}
In \MinmaxSetCovReconf,
we wish to minimize $\cost_\calF(\scrC)$ subject to
$\scrC = ( \calC^\sss, \ldots, \calC^\ttt )$.
For a pair of covers $\calC^\sss$ and $\calC^\ttt$ for $(\calU, \calF)$,
let
$\cost_{\calF}(\calC^\sss \reco \calC^\ttt)$
denote the minimum value of $\cost_\calF(\scrC)$
over all possible reconfiguration sequences $\scrC$
from $\calC^\sss$ to $\calC^\ttt$; namely,
\begin{align}
    \cost_{\calF}(\calC^\sss \reco \calC^\ttt) \coloneq
    \min_{\scrC = ( \calC^\sss, \ldots, \calC^\ttt )}
    \cost_{\calF}(\scrC).
\end{align}
For every $1 \leq c \leq s$,
\prb{Gap$_{c,s}$ \SetCovReconf} requests to distinguish whether
$\cost_{\calF}(\calC^\sss \reco \calC^\ttt) \leq c$ or
$\cost_{\calF}(\calC^\sss \reco \calC^\ttt) > s$.

For the sake of completeness,
we here present a $2$-factor approximation algorithm for \MinmaxSetCovReconf
of \cite{ito2011complexity}:\footnote{
    Similarly, a $2$-factor approximation algorithm can be obtained for
    \MinmaxDomSetReconf and \MinmaxVerCovReconf.
}
\begin{itembox}[l]{\textbf{$2$-factor approximation for \MinmaxSetCovReconf.}}
\begin{algorithmic}[1]
    \LComment{start from $\calC^\sss$.}
    \State insert each set of $\calC^\ttt \setminus \calC^\sss$ into the current cover in any order.
    \State discard each set of $\calC^\sss \setminus \calC^\ttt$ from the current cover in any order.
    \LComment{end with $\calC^\ttt$.}
\end{algorithmic}
\end{itembox}

Our main result is stated below,
whose proof uses a gap-preserving reduction from \LabCovReconf to \MinmaxSetCovReconf
\cite{ohsaka2024gap,lund1994hardness}.

\begin{theorem}\label{thm:SetCover}
\prb{Gap$_{1,2-\frac{1}{\polyloglog N}}$ \SetCovReconf} is $\PSPACE$-complete,
where $N$ is the universe size.
In particular, \MinmaxSetCovReconf
is $\PSPACE$-hard to approximate within a factor of $2-\frac{1}{\polyloglog N}$.
\end{theorem}

\cref{thm:SetCover} along with \cite{ohsaka2024gap} implies that
\MinmaxDomSetReconf is 
$\PSPACE$-hard to approximate within a factor of $2-\frac{1}{\polyloglog N}$,
where $N$ is the number of vertices (see \cref{intro:cor:DominatingSet}).

\begin{proof}[Proof of \cref{thm:SetCover}]
The reduction from \LabCovReconf
to \MinmaxSetCovReconf
is almost the same as that due to \citet{ohsaka2024gap,lund1994hardness}.
Let $(G=(\calV,\calE,\Sigma,\Psi), \asg^\sss, \asg^\ttt)$ be 
an instance of \LabCovReconf
with $N$ vertices and alphabet size $|\Sigma| = \polyloglog N$,
where $\|\asg^\sss\| = \|\asg^\ttt\| = |\calV|$.
Define $B \coloneq \zo^\Sigma$.
For each $\alpha \in \Sigma$ and $S \subseteq \Sigma$,
we construct $\bar{Q_\alpha} \subset B$ and $Q_S \subset B$
according to \cite{ohsaka2024gap} in $2^{\bigO(|\Sigma|)}$ time.
Let $\prec$ be an arbitrary order over $V$.
Create an instance of \MinmaxSetCovReconf as follows.
For each vertex $v \in \calV$ and each value $\alpha \in \Sigma$,
we define $S_{v,\alpha} \subset \calE \times B$ as
\begin{align}
\label{eq:SetCover:S}
    S_{v,\alpha} \coloneq
    \left( \bigcup_{e=(v,w) \in \calE : v \prec w} \{e\} \times \bar{Q_\alpha} \right)
    \cup
    \left( \bigcup_{e=(v,w) \in \calE : v \succ w} \{e\} \times Q_{\pi_e(\alpha)} \right),
\end{align}
where 
$\pi_e(\alpha) \coloneq \{ \beta \in \Sigma \mid \psi_e(\alpha,\beta) = 1\}$.
Then, a set system $(\calU, \calF)$ is defined as
\begin{align}
    \calU \coloneq \calE \times B \text{ and }
    \calF \coloneq \Bigl\{ S_{v,\alpha} \Bigm| v \in \calV, \alpha \in \Sigma \Bigr\}.
\end{align}
For a satisfying multi-assignment $\asg \colon \calV \to 2^\Sigma$ for $G$ with $\|\asg\| = |\calV|$,\footnote{
    In other words, each $f(v)$ is a singleton.
}
we associate it with a subfamily $\calC_\asg \subset \calF$ such that
\begin{align}
\label{eq:SetCover:Casgmt}
    \calC_\asg \coloneq \Bigl\{ S_{v,\alpha} \Bigm| v \in \calV, \alpha \in \asg(v) \Bigr\},
\end{align}
which is a minimum cover for $(\calU, \calF)$ \cite{ohsaka2024gap};
i.e., $|\calC_\asg| = |\calV| = \opt(\calF)$.
Constructing minimum covers $\calC^\sss$ from $\asg^\sss$ and $\calC^\ttt$ from $\asg^\ttt$
by \cref{eq:SetCover:Casgmt},
we obtain an instance
$((\calU,\calF), \calC^\sss, \calC^\ttt)$
of \MinmaxSetCovReconf.
This completes the description of the reduction.

Here, we will show that
\begin{align}
\label{eq:SetCover:minlab-cost}
    \minlab_G(\asg^\sss \reco \asg^\ttt) = \cost_\calF(\calC^\sss \reco \calC^\ttt),
\end{align}
which implies the completeness and soundness;
for this, we use the following lemma \cite{ohsaka2024gap}.
\begin{lemma}[\protect{\cite[Observation~4.4; Claim~4.7]{ohsaka2024gap}}]
\label{lem:SetCover}
Let
$\asg \colon \calV \to 2^\Sigma$ be a multi-assignment and
$\calC \subseteq \calF$ be a subfamily such that
for any $v \in \calV$ and $\alpha \in \Sigma$,
$\alpha \in \asg(v)$ if and only if $S_{v,\alpha} \in \calC$.
Then, 
$\asg$ satisfies an edge $e = (v,w) \in \calE$
if and only if
$\calC$ covers $\{e\} \times B$.
In particular,
$\asg$ satisfies $G$ if and only if $\calC$ covers $\calE \times B$.
Moreover, it holds that $\|\asg\| = |\calC|$.
\end{lemma}\noindent
We first show that
$\minlab_G(\asg^\sss \reco \asg^\ttt) \geq \cost_\calF(\calC^\sss \reco \calC^\ttt)$.
For any reconfiguration multi-assignment sequence
$\sqasg = (\asg^{(1)}, \ldots, \asg^{(T)})$
from $\asg^\sss$ to $\asg^\ttt$
such that $\|\sqasg\|_{\max} = \minlab_G(\asg^\sss \reco \asg^\ttt)$,
we can construct a reconfiguration sequence
$\scrC = (\calC_{\asg^{(1)}}, \ldots, \calC_{\asg^{(T)}})$
from $\calC^\sss$ to $\calC^\ttt$ by \cref{eq:SetCover:Casgmt}.
By \cref{lem:SetCover},
each $\calC_{\asg^{(t)}}$ covers $\calU$;
thus, $\scrC$ is a valid reconfiguration sequence from $\calC^\sss$ to $\calC^\ttt$.
Moreover,
$
    \cost_\calF(\calC^\sss \reco \calC^\ttt)
    \leq \cost_\calF(\scrC)
    = \|\sqasg\|_{\max}
    = \minlab_G(\asg^\sss \reco \asg^\ttt)
$, as desired.
We then show that 
$\minlab_G(\asg^\sss \reco \asg^\ttt) \leq \cost_\calF(\calC^\sss \reco \calC^\ttt)$.
For any reconfiguration sequence
$\scrC = (\calC^{(1)}, \ldots, \calC^{(T)})$
from $\calC^\sss$ to $\calC^\ttt$
such that $\cost_\calF(\scrC) = \cost_\calF(\calC^\sss \reco \calC^\ttt)$,
we can construct a sequence 
$\sqasg = (\asg^{(1)}, \ldots, \asg^{(t)})$
of multi-assignments
such that $\asg^{(t)} \colon \calV \to 2^{\Sigma}$ is defined as follows:
\begin{align}
    \asg^{(t)}(v) \coloneq \Bigl\{ \alpha \in \Sigma \Bigm| S_{v,\alpha} \in \calC^{(t)} \Bigr\}
    \text{ for all } v \in \calV.
\end{align}
By \cref{lem:SetCover},
each $\asg^{(t)}$ satisfies $G$;
thus, $\sqasg$ is a valid reconfiguration multi-assignment sequence from $\asg^\sss$ to $\asg^\ttt$.
Moreover,
$
    \minlab_G(\asg^\sss \reco \asg^\ttt)
    \leq \|\sqasg\|_{\max}
    = \cost_\calF(\scrC)
    = \cost_\calF(\calC^\sss \reco \calC^\ttt)
$, which completes the proof of \cref{eq:SetCover:minlab-cost}.

Since $|\Sigma| = \polyloglog N$,
the reduction takes polynomial time in $N$, and
it holds that $|\calU| = |\calE \times B| = \bigO(N^2 \cdot 2^{\polyloglog N}) = \bigO(N^3)$;
i.e., $N = \Omega(|\calU|^{\frac{1}{3}})$.
By \cref{thm:LabelCover},
\prb{Gap$_{1,2-\frac{1}{\polyloglog N}}$ \LabCovReconf} with alphabet size $\polyloglog N$
is $\PSPACE$-complete; thus,
\prb{Gap$_{1,2-\frac{1}{\polyloglog |\calU|}}$ \SetCovReconf}
is $\PSPACE$-complete as well,
accomplishing the proof.
\end{proof}

\subsection{$\PSPACE$-hardness of Approximation for \VerCovReconf}
\label{subsec:appl:VertexCover}

We conclude this section with a similar inapproximability result for
\MinmaxVerCovReconf on $\bigO(1)$-uniform hypergraphs.
\ifthenelse{\boolean{FULL}}{%
A \emph{vertex cover} of a hypergraph $H=(\calV,\calE)$ is
a vertex set $\calC \subseteq \calV$ that intersects every hyperedge in $\calE$;
i.e., $e \cap \calC \neq \emptyset$ for every $e \in \calE$.
For any pair of vertex covers $\calC^\sss$ and $\calC^\ttt$ of $H$,
a \emph{reconfiguration sequence from $\calC^\sss$ to $\calC^\ttt$} is a sequence
$\scrC = (\calC^{(1)}, \ldots, \calC^{(T)})$ of vertex covers such that
$\calC^{(1)} = \calC^\sss$,
$\calC^{(T)} = \calC^\ttt$, and
$|\calC^{(t)} \triangle \calC^{(t+1)}| \leq 1$ for all $t$.
Denote by $\beta(H)$ the size of the minimum vertex cover of $H$.
For a reconfiguration sequence
$\scrC = (\calC^{(1)}, \ldots \calC^{(T)})$,
its \emph{cost} is defined as the maximum value of
$\frac{|\calC^{(t)}|}{\beta(H)+1}$
over all $\calC^{(t)}$'s in $\scrC$; namely,
\begin{align}
    \cost_H(\scrC) \coloneq \max_{\calC^{(t)} \in \scrC} \frac{|\calC^{(t)}|}{\beta(H)+1}.
\end{align}
In the \MinmaxVerCovReconf problem,
for a hypergraph $H$ and its two vertex covers $\calC^\sss$ and $\calC^\ttt$,
we wish to minimize $\cost_H(\scrC)$
subject to $\scrC = (\calC^\sss, \ldots, \calC^\ttt)$.
Let $\cost_H(\calC^\sss \reco \calC^\ttt)$ denote the minimum value of 
$\cost_H(\scrC)$ over all possible reconfiguration sequences $\scrC$
from $\calC^\sss$ to $\calC^\ttt$; namely,
\begin{align}
    \cost_H(\calC^\sss \reco \calC^\ttt)
    \coloneq \min_{\scrC = (\calC^\sss, \ldots, \calC^\ttt)} \cost_H(\scrC).
\end{align}
For every $1 \leq c \leq s$,
\prb{Gap$_{c,s}$ \VerCovReconf}
requires to distinguish
whether $\cost_H(\calC^\sss \reco \calC^\ttt) \leq c$ or
$\cost_H(\calC^\sss \reco \calC^\ttt) > s$.

}{
\MinmaxVerCovReconf is defined analogously to \MinmaxSetCovReconf; refer to \cref{app:appl} for the formal definition.
}
Our inapproximability result is shown below,
whose proof reuses the reduction of \cref{thm:SetCover}.

\ifthenelse{\boolean{FULL}}{
\begin{theorem}
}{
\begin{theorem}[$\ast$]
}
\label{thm:VertexCover}
For any constant $\epsilon \in (0,1)$,
\prb{Gap$_{1,2-\epsilon}$ \VerCovReconf}
on $\poly(\epsilon^{-1})$-uniform hypergraphs
is $\PSPACE$-complete.
In particular, \MinmaxVerCovReconf
on $\poly(\epsilon^{-1})$-uniform hypergraphs
is $\PSPACE$-hard to approximate within a factor of $2-\epsilon$.
\end{theorem}
\ifthenelse{\boolean{FULL}}{
\begin{proof}
We build an ``inverted index'' of
\prb{Gap$_{1,2-\epsilon}$ \SetCovReconf} of \cref{thm:SetCover}.
Let $(G=(\calV,\calE,\Sigma,\Psi), \asg^\sss, \asg^\ttt)$
be an instance of \LabCovReconf
with $N$ variables and alphabet size $|\Sigma| = \poly(\epsilon^{-1})$,
where $\|\asg^\sss\| = \|\asg^\ttt\| = |\calV|$.
Define $B \coloneq \zo^\Sigma$ and $S_{v,\alpha}$'s by \cref{eq:SetCover:S}.
Construct then a hypergraph $H = (\calW,\calF)$ as follows:
\begin{align}
    \calW & \coloneq \calV \times \Sigma, \\
    \calF & \coloneq \Bigl\{T_{e,\vec{q}} \subset \calW \Bigm| (e,\vec{q}) \in \calE \times B \Bigr\},
        \quad \text{where} \\
    T_{e,\vec{q}} &
        \coloneq \Bigl\{ (v,\alpha) \in \calW \Bigm| (e,\vec{q}) \in S_{v,\alpha} \Bigr\}
        \quad \text{for all } (e,\vec{q}) \in \calE \times B.
\end{align}
Note that the size of hyperedge $T_{e,\vec{q}}$ is 
\begin{align}
\begin{aligned}
    |T_{e,\vec{q}}|
    & = \left|\Bigl\{ (v,\alpha) \in \calW \Bigm| (e,\vec{q}) \in S_{v,\alpha} \Bigr\}\right| \\
    &   \leq \left|\Bigl\{
            (v,\alpha) \in \calV \times \Sigma \Bigm| v \in e
        \Bigr\}\right| \\
    & \leq 2|\Sigma| = \poly(\epsilon^{-1}).
\end{aligned}
\end{align}
To make $H$ into $2|\Sigma|$-uniform, we simply augment each hyperedge $T_{e,\vec{q}}$
with a set of fresh $2|\Sigma| - |T_{e,\vec{q}}|$ vertices.
For a satisfying multi-assignment $\asg \colon \calV \to 2^\Sigma$ for $G$ with $\|\asg\| = |\calV|$,
we associate with it a vertex set $\calC_\asg \subset \calW$ such that
\begin{align}
\label{eq:VertexCover:Casgmt}
    \calC_\asg \coloneq \Bigl\{(v,\alpha) \in \calV \times \Sigma \Bigm| v \in \calV, \alpha \in \asg(v) \Bigr\},
\end{align}
which is a minimum vertex cover of $H$ (i.e., $|\calC_\asg| = |\calV| = \beta(H)$),
shown similarly to the proof of \cref{thm:SetCover}.
Constructing minimum vertex covers $\calC^\sss$ from $\asg^\sss$ and $\calC^\ttt$ from $\asg^\ttt$
by \cref{eq:VertexCover:Casgmt},
we obtain an instance $(H, \calC^\sss, \calC^\ttt)$ of \MinmaxVerCovReconf
on a $\poly(\epsilon^{-1})$-uniform hypergraph.
This completes the description of the reduction.

Similarly to the proof \cref{thm:SetCover}, we can show that
\begin{align}
    \minlab_G(\asg^\sss \reco \asg^\ttt) = \cost_H(\calC^\sss \reco \calC^\ttt),
\end{align}
implying the completeness and soundness.
By \cref{thm:LabelCover},
\prb{Gap$_{1,2-\epsilon}$ \LabCovReconf}
with alphabet size $\poly(\epsilon^{-1})$ is $\PSPACE$-complete; thus,
\prb{Gap$_{1,2-\epsilon}$ \VerCovReconf}
on $\poly(\epsilon^{-1})$-uniform hypergraphs is $\PSPACE$-complete as well,
which completes the proof.

\end{proof}
}{
}

\printbibliography

\end{document}